\documentclass[12pt]{iopart}
\usepackage{iopams}
\usepackage{amsmath}
\usepackage{amssymb}
\usepackage{bbm}
\usepackage{amsthm}
\usepackage[pdftex]{color}
\usepackage{float}
\usepackage[caption = false]{subfig}
\usepackage{graphicx}% Include figure files
\usepackage{dcolumn} % Align table columns on decimal point
\usepackage{bm} % bold math
\usepackage{epic}
\usepackage{longtable}
\usepackage{ulem}   % to

\newtheorem{theorem}{Theorem}

\newcommand{\1}{\mathbbm{1}}

\newcommand{\A}{\mathcal{A}}
\newcommand{\ket}{\rangle}
\newcommand{\bra}{\langle}

\begin{document}

\title[Generalized Shortcuts to Adiabaticity]{Generalized Shortcuts to Adiabaticity and Enhanced Robustness Against Decoherence}

\author{Alan C. Santos \& Marcelo S. Sarandy}

\address{Instituto de F\'{i}sica, Universidade Federal Fluminense, Av. Gal. Milton Tavares de Souza s/n, Gragoat\'{a}, 24210-346 Niter\'{o}i, Rio de Janeiro, Brazil}
\ead{msarandy@id.uff.br}
\vspace{10pt}
\begin{indented}
\item[]August 2017
\end{indented}

\begin{abstract}
Shortcuts to adiabaticity provide a general approach to mimic adiabatic quantum processes via arbitrarily 
fast evolutions in Hilbert space. For these counter-diabatic evolutions, higher speed comes at higher energy 
cost. Here, the counter-diabatic theory is employed as a minimal energy demanding scheme 
for speeding up adiabatic tasks. As a by-product, we show that this approach can be used to obtain infinite 
classes of transitionless models, including time-independent Hamiltonians under certain conditions over the 
eigenstates of the original  Hamiltonian. We apply these 
results to investigate shortcuts to adiabaticity in decohering environments by introducing the requirement 
of a fixed energy resource. In this scenario, we show that generalized transitionless evolutions can be more 
robust against decoherence than their adiabatic counterparts. We illustrate this enhanced robustness both 
for the Landau-Zener model and for quantum gate Hamiltonians.
\end{abstract}

\section{Introduction}
The adiabatic theorem~\cite{Born:28,Kato:50,Messiah:book,Teufel:03} constitutes a successful strategy for 
eigenstate tracking in quantum information and quantum control (see, e.g., Ref.~\cite{Jing:16}). 
It states that a system initially prepared in an eigenstate of a time-dependent Hamiltonian $H(t)$ 
will evolve to the corresponding instantaneous eigenstate at a later time $T$, provided that $H(t)$ 
varies smoothly and that $T$ is much larger than a power of the relevant minimal inverse 
energy gap.
It is worth highlighting that the validity conditions of the adiabatic approximation have been 
revisited (see, e.g. Refs~\cite{Marzlin:04,Sarandy:04,Tong:05,Suter,Jianda-Wu:08,Amin:09}),  
which has implied in enhanced formulations of the adiabatic theorem~\cite{Ambainis:04,Jansen:07} 
(see also Ref.~\cite{Tameem-Lidar} for a recent review).
In a real open-system scenario, 
the performance of the adiabatic dynamics is upper bounded by the competition between the adiabatic 
time scale, which is favored by a slow evolution, and the typically short decoherence time scales. 
This interplay provides an optimal time scale for adiabatic processes in decohering environments~\cite{Sarandy:05,Sarandy:05-2}.  

The adiabatic dynamics can be reproduced by generalized transitionless evolutions obtained via shortcuts 
to adiabaticity~\cite{Demirplak:03,Demirplak:05,Berry:09}. Such accelerated processes allow us to derive  
an exact adiabatic evolution at an arbitrary finite time. Shortcuts to adiabaticity have been used to speed 
up adiabatic processes in a number of applications, e.g. tracking of many-body systems across quantum 
phase transitions~\cite{Adolfo:13,Saberi:14,Hatomura:17}, quantum gate Hamiltonians~\cite{SciRep:15,PRA:16,Front:16}, 
heat engines in quantum thermodynamics~\cite{Adolfo:16}, among others (e.g., Refs.~\cite{Stefanatos:14,Lu:14,Deffner:16,XiaPRA:14,Song:16}). 
The robustness of such transitionless evolutions has recently been studied through different experimental 
architectures, as nitrogen-vacancy setups~\cite{Liang:16}, trapped ions~\cite{An:16}, atoms in cavities~\cite{Xia:16,Yi:17}, 
nuclear magnetic resonance (NMR)~\cite{Ramanathan:16}, and optomechanics~\cite{Zhang:16}.
Naturally, the speed of the evolution is constrained by the energy cost of the implementation, with faster 
evolutions being more energy demanding~\cite{SciRep:15}. By providing identical energy resources at a finite 
evolution time $\tau$,
a fundamental problem is then whether shortcuts to adiabaticity can provide a more efficient performance 
in terms of fidelity than their adiabatic counterparts by adjusting its pace within the decoherence 
time scales. We address this question by considering a general counter-diabatic theory~\cite{Torrontegui:13}, 
which is here optimized for a minimum energy consumption. As a by-product, we apply this 
general approach to obtain infinite classes of transitionless models, including 
time-independent Hamiltonians under certain conditions over the eigenstates of the 
original (adiabatic) Hamiltonian. Concerning robustness against decoherence, we consider 
Markovian open systems and impose fixed energy resources. This is a key point, since unlimited energy 
provides arbitrarily fast dynamics already for adiabatic evolution, through an arbitrarily large gap 
between the ground and first excited states. It is then shown that a supremacy of the
counter-diabatic dynamics can always be achieved by adjusting the evolution rate.  
This is illustrated in the Landau-Zener model and in quantum gate Hamiltonians.  
%%%%%%%%%%%%%%%%%%%%%%%%%%%%%%%%%

\section{Generalized transitionless dynamics theory and energy cost} 
The starting point for the counter-diabatic theory is the evolution operator $U\left( t\right)$, 
which can be defined as (see, e.g., Ref.~\cite{Torrontegui:13})
\begin{equation}
U\left( t\right) =\sum\nolimits_{n}e^{i \int_{0}^{t} \theta _{n}\left( \xi \right)
	d \xi } \vert n_t \ket \bra n_{0} \vert \text{ ,}  \label{operadorU}
\end{equation}%
where $\theta _{n}\left( t\right)$ is a set of arbitrary real phases~\cite{Muga:10,Chen:11} and 
$\{ \vert n_t \ket = \vert n (t) \ket \}$ is the set of eigenstates of the original (adiabatic) 
Hamiltonian $H_{0} ( t )$. Let us assume that the quantum system is initially prepared 
in an specific eigenstate $\vert k_0\ket$ of $H_{0} ( t )$, 
namely, $\vert \psi (0) \ket = \vert k_0\ket$. Then, the Hamiltonian 
$H_{ \text{\text{SA}} }\left( t\right) =-i  U\left( t\right) \dot{U}^{\dag }\left(t\right)$, which denotes 
the {\it shortcut to the adiabatic} Hamiltonian $H_0(t)$, 
evolves the system to its instantaneous eigenlevel 
$\vert \psi (t) \ket = e^{i \int_{0}^{t} \theta _{k}\left( \xi \right) d \xi } \vert k_t\ket$. 
Explicitly, we write $H_{ \text{\text{SA}} }\left( t\right)$ as ($\hbar = 1$)
\begin{equation}
H_{ \text{\text{SA}} }\left( t\right) =i  \sum\nolimits_{n} \left( |\dot{n}_t\rangle
\langle n_t |+i\theta _{n}\left( t\right) |n_t
\rangle \langle n_t | \right) \text{ .} 
\label{HSA}
\end{equation}
The functions $\theta _{n}\left( t\right)$ have originally been identified with the adiabatic 
phase  $\theta_n(t) = - E_n(t) + i \langle n_t |\dot{n}_t\rangle$~ \cite{Berry:84}, which 
exactly mimics an adiabatic evolution. In this particular case, we can write 
$H_{\text{SA}}(t)$ as $H_0(t) + H_{\text{CD}}(t)$, where 
$H_0(t) = \sum_n E_n(t) |n_t\rangle \langle n_t|$ is the Hamiltonian that drives the adiabatic 
dynamics and $H_{\text{CD}}\left( t\right) =i  \sum_{n}\left( \left\vert 
\dot{n}_t\right\rangle \left\langle n_t\right\vert +\left\langle \dot{n}_t|n_t\right\rangle 
\left\vert n_t\right\rangle \left\langle n_t\right\vert
\right)$ is the {\it counter-diabatic} Hamiltonian.

\subsection{Quantum phases in optimal transitionless quantum driving}

There is a number of situations for which we need not exactly mimic an adiabatic 
process, but only assure that the system is kept in an instantaneous eigenstate (independently 
of its associated quantum phase)~\cite{SciRep:15,PRA:16,Front:16,Adolfo:16,Stefanatos:14,Lu:14,
	Deffner:16,Liang:16,An:16,Xia:16,Zhang:16,Marcela:14,Chen:10}. 
This generalized dynamics in terms of arbitrary phases $\theta_n(t)$ will be denoted 
as a {\it transitionless} evolution. Now, we will show that $\theta_n(t)$ can be nontrivially 
optimized in transitionless evolutions both in terms of energy cost and robustness against 
decoherence effects. In this direction, we adopt as a measure of energy cost the average 
Hilbert-Schmidt norm of the Hamiltonian throughout the 
evolution, which is 
given by~\cite{Zheng:16,SciRep:15,Campbell-Deffner:17,Nathan:14}
\begin{equation}
\Sigma  \left( \tau \right) = \frac{1}{\tau }\int_{0}^{\tau }||H(t)|| \text{ }dt = \frac{1}{\tau }\int_{0}^{\tau }\sqrt{\text{Tr}%
	\left[ {H}^{2}\left( t\right) \right] }\text{ }dt \text{ \ ,} \label{costGeneric}
\end{equation}
where $\tau$ denotes the total evolution time.
The energy cost $\Sigma  \left( \tau \right)$ aims at identifying changes in the energy coupling constants and gap structure, 
which typically account for the effort of speeding up adiabatic processes. It is a well-defined measure for 
finite-dimensional Hamiltonians exhibiting non-degeneracies in their energy spectra. Such Hamiltonians describe 
the quantum systems within the scope of this work. Therefore, Eq.~(\ref{costGeneric}) is applicable, e.g., for  
generic systems composed of a finite number of quantum bits (qubits) under magnetic or electric fields 
(see Refs.~\cite{Nathan:14,Zheng:16} for similar cost measures). Note that Eq.~(\ref{costGeneric}) is 
non-invariant with respect to a change of the zero energy offset. However, by adopting a fixed 
reference frame, it can be used to quantify the energy cost involved in attempts of  
accelerating the adiabatic path, either via an increasing of the energy gap in the adiabatic approach or via 
a reduction of $\tau$ by adjusting the relevant energy couplings in the counter-diabatic theory. In addition, as $\tau$ can be set by the quantum speed limit~\cite{Deffner:13}, Eq.~(\ref{costGeneric}) 
allows us to establish a trade-off between speed and energy cost for an arbitrary 
dynamics~\cite{SciRep:15,Campbell-Deffner:17}. For instance, in NMR experimental setups, 
the  quantity $||H(s)||$ represents how intense a magnetic field $\vec{B}(s)$ is expected to be in order to control the speed of such a dynamics. 

Now, let us discuss how, for a fixed time $\tau$, the energy cost in a transitionless evolution can be minimized by a suitable choice of the 
arbitrary parameters $\theta_{n}\left( t\right)$. Remarkably, this optimization can be analytically 
derived, which is established by Theorem~\ref{TheoOptmEner} below. Its derivation is provided 
in \ref{Theorem1}. 

\begin{theorem} \label{TheoOptmEner}
	Consider a closed quantum system under adiabatic evolution governed by a Hamiltonian 
	$H_0\left( t\right) $. The energy cost to implement its generalized transitionless counterpart,  
	driven by the Hamiltonian $H_{\text{SA}}(t)$, can be minimized by setting  
	\begin{equation}
	\theta _{n}\left( t\right) = \theta ^{\text{min}}_{n}\left( t\right) = -i\langle \dot{n}_{t}|n_{t}\rangle \text{ .}
	\label{OptTeta}
	\end{equation}
\end{theorem}

In particular, for any evolution such that the quantum parallel-transport condition is 
verified~\cite{Berry:84}, the energy cost to implement a transitionless evolutions is 
always optimized by choosing $\theta_{n}^{min}\left( t\right)=0$. 
This approach is useful for providing both realistic and energetically optimal Hamiltonians in several physical scenarios. 
For example, by considering nuclear spins driven by a magnetic field $\vec{B}(t)$ in a nuclear magnetic resonance setup, the energy 
cost can be optimized by adjusting $\theta_{n}^{min}\left( t\right)$ such that the magnitude $B(t)$ of the magnetic field is reduced, 
since $||H(t)|| \propto B(t) $.

As a by-product, 
the generalized counter-diabatic theory can be used as a tool to yield time-independent 
Hamiltonians for transitionless evolutions. In general, the Hamiltonian $H_{\text{SA}}(t)$ has its form 
constrained both by the choice of the phases $\theta_{n} \left( t\right) $ and by eigenstates 
of the adiabatic Hamiltonian $H_0(t)$. Thus, we can delineate under what conditions we can choose 
the set $\{\theta_{n} \left( t\right)\}$ in order to obtain a time-independent Hamiltonian for a transitionless evolution.
To answer this question, we impose $\dot{H}_{ \text{\text{SA}} }\left( t\right)=0$ considering 
arbitrary phases $\theta_{n}\left( t\right)$. This leads to Theorem~\ref{TheoTimeIndep} below.  
Its derivation is provided in \ref{Theorem2}.

\begin{theorem} \label{TheoTimeIndep}
	Let $H_0\left( t\right) $ be a discrete quantum Hamiltonian, with $\{|m_{t}\rangle\}$ denoting its 
	set of instantaneous eigenstates. If $\{|m_{t}\rangle\}$ satisfies $\langle k_t |\dot{m}_{t}\rangle = c_{km}$, 
	with $c_{km}$ complex constants $\forall k,m$, then a family of time-independent 
	Hamiltonians $H^{\{\theta\}}$ for generalized transitionless evolutions can be defined by setting 
	$\theta_{m}\left( t\right) = \theta$, with $\theta$ a single arbitrary real constant $\forall m$.
\end{theorem}
%%%%%%%%%%%%%%%%%%%%%%%%%%%%%%%%%%%%%%%%%%%%
\subsection{Transitionless dynamics under decoherence}

Theorems \ref{TheoOptmEner} and 
\ref{TheoTimeIndep} ensure both an energetically optimal counter-diabatic evolution and 
families of possible time-independent transitionless Hamiltonians. A rather important point for the 
generalized counter-diabatic theory is whether it is robust against decoherence. 
The robustness of the counter-diabatic dynamics and inverse engineering schemes has recently been
considered in the literature~\cite{Hao:16,Chen:16,Liu:17,Jing:13}. Here, in order to provide a 
comparison between adiabatic and generalized counter-diabatic dynamics, 
we will require {\it identical energy resources} for each implementation. 
More specifically, we will consider the performance of transitionless evolutions in open 
systems described by convolutionless master equations given by
\begin{equation}
d_s \rho \left( s\right) =-i \tau\left[ H_{\text{SA}}\left(
s\right) ,\rho \left( s\right) \right] +\tau \mathcal{L}_{i}\left[ \rho
\left( s\right) \right] \text{ ,} \label{parLindEq}
\end{equation}
where $\mathcal{L}_{i}\left[ \rho
\left( s\right) \right]$ describes the decohering contribution to the quantum dynamics, 
which is parametrized by the normalized time $s=t/\tau $, with $\tau$ the total time of 
evolution and $0 \le s \le 1$. For Markovian evolution~\cite{Lindblad:76,Petruccione:Book}, we have 
$\mathcal{L}_{i}\left[ \rho\left( s\right)\right] = \frac{1}{2} \sum\nolimits_{i}\gamma ^{2} _{i}\left( s\right) 
[ 2 L_{i}\left( s\right) \rho
\left( s\right) L_{i}^{\dag }\left( s\right) - \{
L_{i}^{\dag }\left( s\right) L_{i}\left( s\right) ,\rho
\left( s\right) \} ]$, with $L_{i}(s)$ denoting Lindblad operators and $\gamma_i(s)$ (positive) 
decoherence rates. Here, we will consider, as an illustration, Lindblad operators for generalized amplitude 
damping (GAD) in the eigenbasis of the Hamiltonian, which reads
\begin{eqnarray}
L_{\pm}^{\text{GAD}}\left( s\right) &=& U^{ \dagger }(s) \sigma_{\pm} U(s) \text{ ,} \label{Lind-GAD} 
\end{eqnarray}
where  $U(s)$ is the unitary operator that diagonalizes the Hamiltonian and 
$\sigma_{\pm} = (\sigma_{x} \mp i\sigma_{y})/2$, with $\{ \sigma_x, \sigma_y, \sigma_z \}$ 
denoting Pauli matrices. The GAD channel describes dissipation to an environment at 
finite temperature. Its decoherence rates $\gamma_+$ and $\gamma_-$ are given by~\cite{Srikanth:08,Cafaro:14} 
$\gamma_+ =  \sqrt{\gamma_0 N_{\text{th}}}$ and $\gamma_- =  \sqrt{\gamma_0 \left( N_{\text{th}} + 1 \right)}$, 
where $\gamma_0$ is the spontaneous emission rate and $N_{\text{th}}$ is the Planck distribution that gives 
the number of thermal photons at a fixed frequency. For simplicity, we adjust
the temperature such that $N_{\text{th}}=1/2$ and define $\gamma_0 \equiv \alpha \omega_r$, with 
$\alpha$ a dimensionless parameter and $\omega_r$ some relevant frequency associated with the 
quantum system. Then, we obtain 
$\gamma_{+} = \sqrt{\alpha \omega_{r}/2}$ and 
$\gamma_{-} = \sqrt{3}\gamma_{+}$. In addition to GAD, we will also consider  
dephasing in the instantaneous Hamiltonian eigenbasis, whose Lindblad operator reads 
\begin{equation}
L_{\text{d}}\left( s\right) = U^{ \dagger }(s) \sigma_{z} U(s) \text{ ,}
\end{equation}
with decoherence rate given by 
$\gamma_d \equiv \alpha \omega_r$. Both GAD and dephasing are common decohering 
processes in a number of physical realizations~\cite{Petruccione:Book}. They will be used here 
as probes to the counter-diabatic robustness in the open-system realm. 

In this paper we consider that any systematic error due to experimental deviations of fields used to implement the Hamiltonian is negligible. In general, finding an optimal transitionless scheme against arbitrary systematic errors is not a trivial task~\cite{Ruschhaupt:12}. In particular, given a fixed class of error, we can obtain an optimal transitionless dynamics for such a class, but the associated dynamics may be not robust against other classes of systematic errors~ \cite{Ruschhaupt:12,Tseng:14,X-Jing:13}.

%%%%%%%%%%%%%%%%%%%%%%%%%%%%%%%%%%%%%%%%%%
\section{Transitionless dynamics in the Landau-Zener model}
%%%%%%%%%%%%%%%%%%%%%%%%%%%%%%%%%%%%%%%%%%

As a first application, let us consider the dynamics of a 
two-level quantum system, i.e., a qubit, evolving under the Landau-Zener Hamiltonian 
\begin{equation}
H_0^{\text{LZ}}\left( s\right) = - \omega \left[ \sigma _{z}+\tan \vartheta \left( s\right) \sigma _{x}\right],
\label{H-LZ}
\end{equation}
with $\tan \vartheta \left( s\right)$ a dimensionless time-dependent parameter associated with Rabi frequency. This Hamiltonian describes transitions in two-level systems exhibiting anti-crossings 
in its energy spectrum~\cite{Zener:32}. In particular, it can be applied, e.g., to perform adiabatic population transfer in a two-level system driven by a chirped field~\cite{Demirplak:03} and to investigate molecular collision processes \cite{Lee:79}.
The instantaneous ground $|E_{-}\left( s\right) \rangle$ 
and first excited $|E_{+}\left( s\right) \rangle$ states of 
$H_0^{\text{LZ}}\left( s\right)$ are
\begin{eqnarray}
|E_{-}\left( s\right) \rangle  &=&\cos \left[ \frac{\vartheta \left(
	s\right) }{2}\right] |0\rangle +\sin \left[ \frac{\vartheta \left( s\right) 
}{2}\right] |1\rangle \text{ ,}  \label{FundLZ} \\
|E_{+}\left( s\right) \rangle  &=&-\sin \left[ \frac{\vartheta \left(
	s\right) }{2}\right] |0\rangle +\cos \left[ \frac{\vartheta \left( s\right) 
}{2}\right] |1\rangle \text{ .}  \label{ExcLZ}
\end{eqnarray}

The system is initialized in the ground state $|E_{-}\left( 0\right) \rangle =|0\rangle $ 
of $H_0^{\text{LZ}}\left( 0\right)$. By considering a unitary dynamics and a sufficiently 
large total evolution time (adiabatic time), the qubit evolves to the 
instantaneous ground state $|E_{-}\left( s\right) \rangle$ of $H_0^{\text{LZ}}\left( s\right)$. 

\subsection{Optimal transitionless evolution in Landau-Zener model}
In this section we will discuss the generalized transitionless dynamics theory for the Landau-Zener model. 
For optimal energy cost, Eq. (\ref{OptTeta}) establishes 
$\theta _{n}\left( t\right) =\langle d_{s}E_{\pm }\left( s\right)
|E_{\pm }\left( s\right) \rangle =0$ for the states in Eqs. (\ref{FundLZ}) and (\ref{ExcLZ}). 
Therefore, the optimal Hamiltonian ${H}_{\text{SA}}\left( s\right) $ is given by 
${H}_{\text{SA}}\left( s\right) =H_{\text{CD}}^{\text{LZ}}\left( s\right) $, with 
\begin{equation}
H_{\text{CD}}^{\text{LZ}}\left( s\right) = i \sum_{k=\pm }|d_{s}E_{k}\left( s\right)
\rangle \langle E_{k}\left( s\right) |=\frac{d_{s}\vartheta \left( s\right) 
}{2\tau }\sigma _{y} \text{ .}
\label{lz-1}
\end{equation}
From Eq.~(\ref{lz-1}), we can see that $H_0^{\text{LZ}}\left( s\right)$ satisfies the hypotheses of  
Theorem \ref{TheoTimeIndep} if, and only if, we choose the linear interpolation 
$\vartheta (s) = \vartheta _{0}s$. Thus, we adopt this choice 
for simplicity and, consequently, we 
have $H_{\text{SA}}\left( s\right) = (\vartheta _{0}/2\tau ) \sigma _{y}$.
We observe that a complete (avoided) level crossing picture for the Landau-Zener model 
is described by varying the parameter $\tan \vartheta \left( s\right)$ from $-\infty$ to $+\infty$. 
Here, we are taking a narrower range for $\tan \vartheta \left( s\right)$, which simplifies the 
description of the transitionless dynamics for the model.

%%%%%%%%%%%%%%%%%%%%%%%%%%%%%%%%%%%%%%%%%%
\subsection{Energy cost for the Landau-Zener model}
%%%%%%%%%%%%%%%%%%%%%%%%%

Now we will be interested in the performance of the transitionless 
evolution with optimal energy resource so that we impose $\theta_{n}(s)$ as in the 
Theorem \ref{TheoOptmEner}. Thus, by considering adiabatic evolution through 
the Hamiltonian $H_0^{\text{LZ}}\left( s\right)$, Theorem \ref{TheoOptmEner} establishes that the optimal energy resource is 
performed by setting $\theta_{n}(s)=0$.
Considering the energy cost as provided by Eq.~(\ref{costGeneric}), we get
\begin{eqnarray}
\Sigma _{\text{Ad}}\left( \tau \right)  &=&\sqrt{2} |\omega | 
\int_{0}^{1}|\sec [\vartheta \left( s\right) ] |ds\text{ ,} \label{adLZcost} \\
{\Sigma}_{\text{SA}}\left( \tau \right)  &=&\int_{0}^{1}\frac{%
	| d_{s}\vartheta \left( s\right) |}{\sqrt{2}\tau }ds=\frac{| \vartheta \left(
	1\right) |}{\sqrt{2}\tau } \label{SuperOptLZcost} \text{ .}
\end{eqnarray}
where $\Sigma _{\text{Ad}}$ and ${\Sigma}_{\text{SA}}$ are the energy costs for 
the adiabatic and optimal shortcut Hamiltonians, 
respectively. Remarkably, from Eqs. (\ref{adLZcost}) and (\ref{SuperOptLZcost}), 
it follows that the energy cost for the counter-diabatic Landau-Zener model is independent 
of the path followed by the system on the Bloch sphere, while its adiabatic counterpart depends 
on it. In particular, for obtaining ${\Sigma}_{\text{SA}}(\tau)$, we have used  
$\tan \vartheta \left( 0\right) =0$ and, therefore, $\vartheta
\left( 0\right) =0$. 
Note that there is a range of values for $\tau$ for which the energy cost of the 
generalized transitionless dynamics is less than its adiabatic version. 
Indeed, by evaluating the relation between $\Sigma _{\text{Ad}}\left( \tau \right)$ and ${\Sigma}_{\text{SA}}\left( \tau \right) $ we get
\begin{equation}
{\mathcal{R}} \left( \tau \right) = \frac{\Sigma _{\text{Ad}}\left( \tau \right)}{{\Sigma}_{\text{SA}}\left( \tau \right)} = | \omega | \tau \frac{ 2 \int_{0}^{1}|\sec [\vartheta \left( s\right) ]|ds }{ |\vartheta \left(
	1\right) | } . \label{Rfunctiontilde}
\end{equation}
By imposing identical energy cost, i.e. ${\mathcal{R}} \left( \tau \right) = 1$, we obtain 
\begin{equation}
|\omega | \tau = \frac{ | \vartheta \left( 1 \right) | }{ 2 \int_{0}^{1} | \sec [\vartheta \left( s\right) ]|ds } \label{constraint-er}.
\end{equation}
Therefore, identical energy cost can be obtained by adjusting $\omega$ according to the total evolution 
time $\tau$, as in Eq.~(\ref{Rfunctiontilde}).

\begin{figure*}[t]
	\centering
	\subfloat[ ]{\includegraphics[scale=0.6]{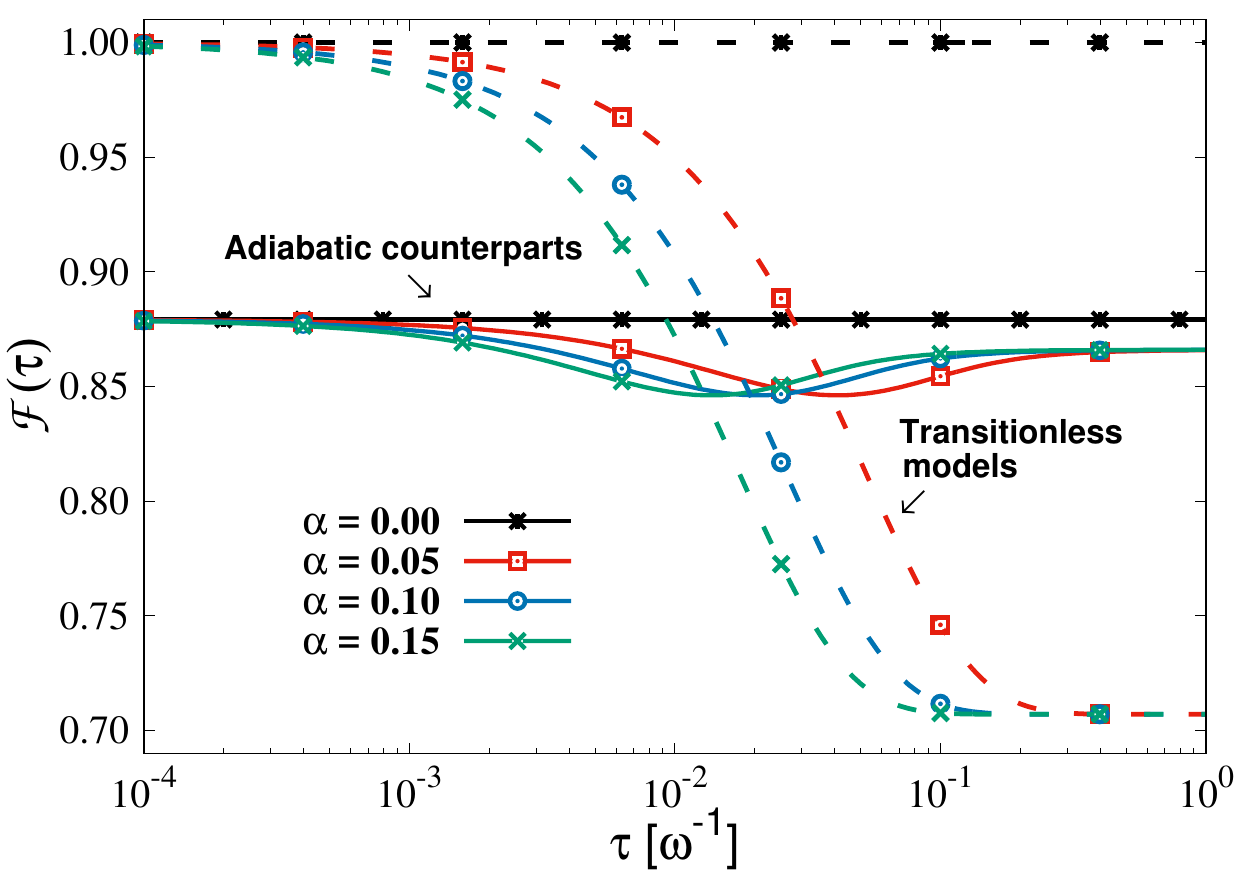} \label{Fig1-a}} \hspace{0.1cm} \subfloat[ ]{\includegraphics[scale=0.6]{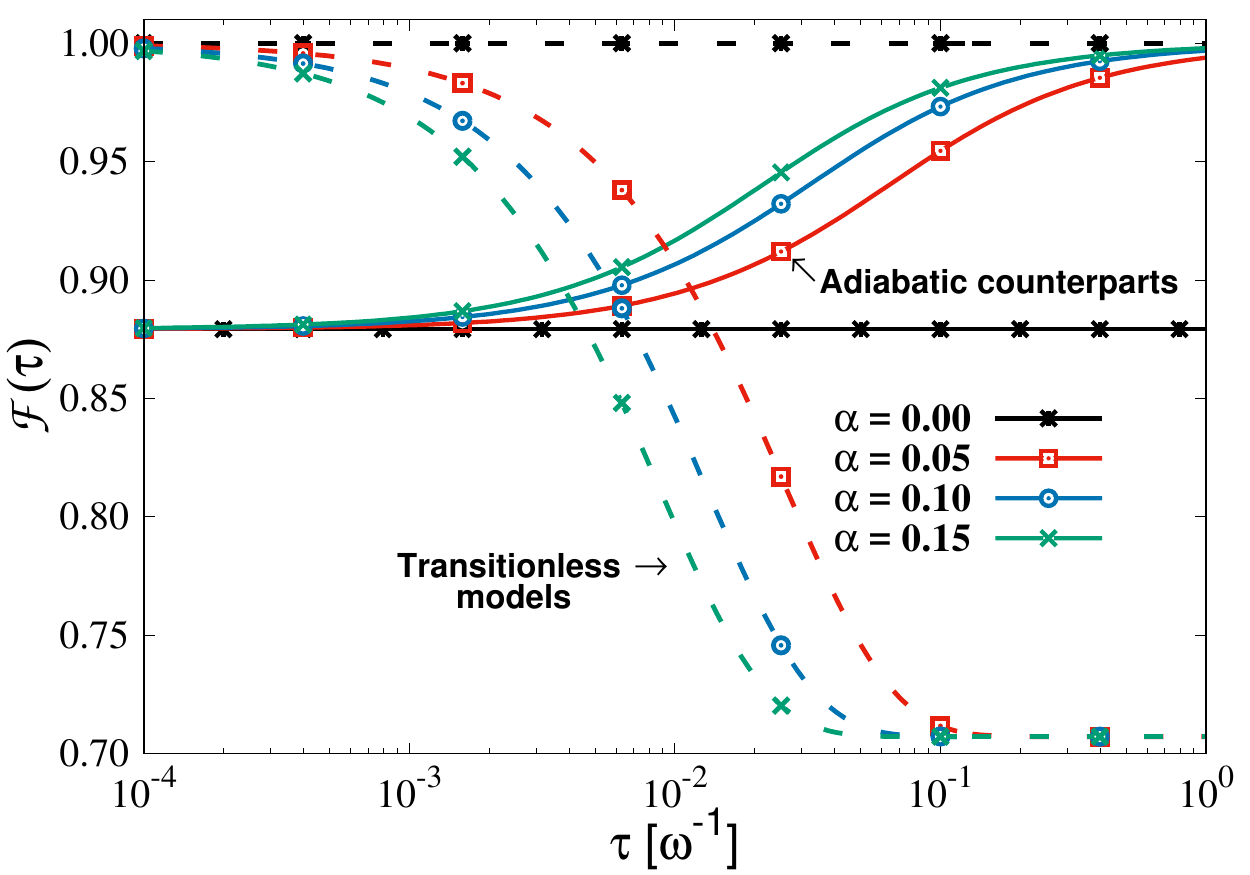}\label{Fig1-b}} 
	\quad
	\subfloat[ ]{\includegraphics[scale=0.6]{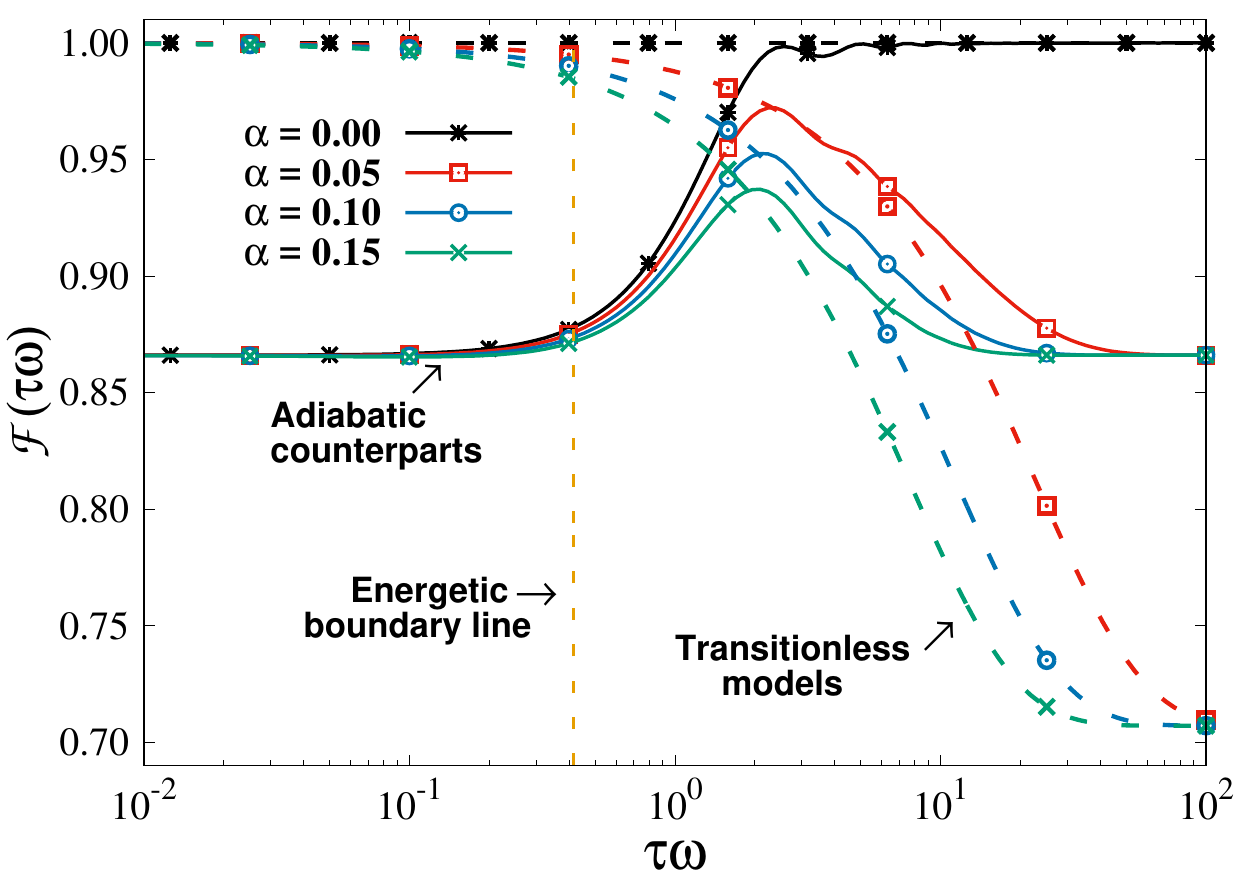}\label{Fig1-c}}\hspace{0.1cm} \subfloat[ ]{\includegraphics[scale=0.6]{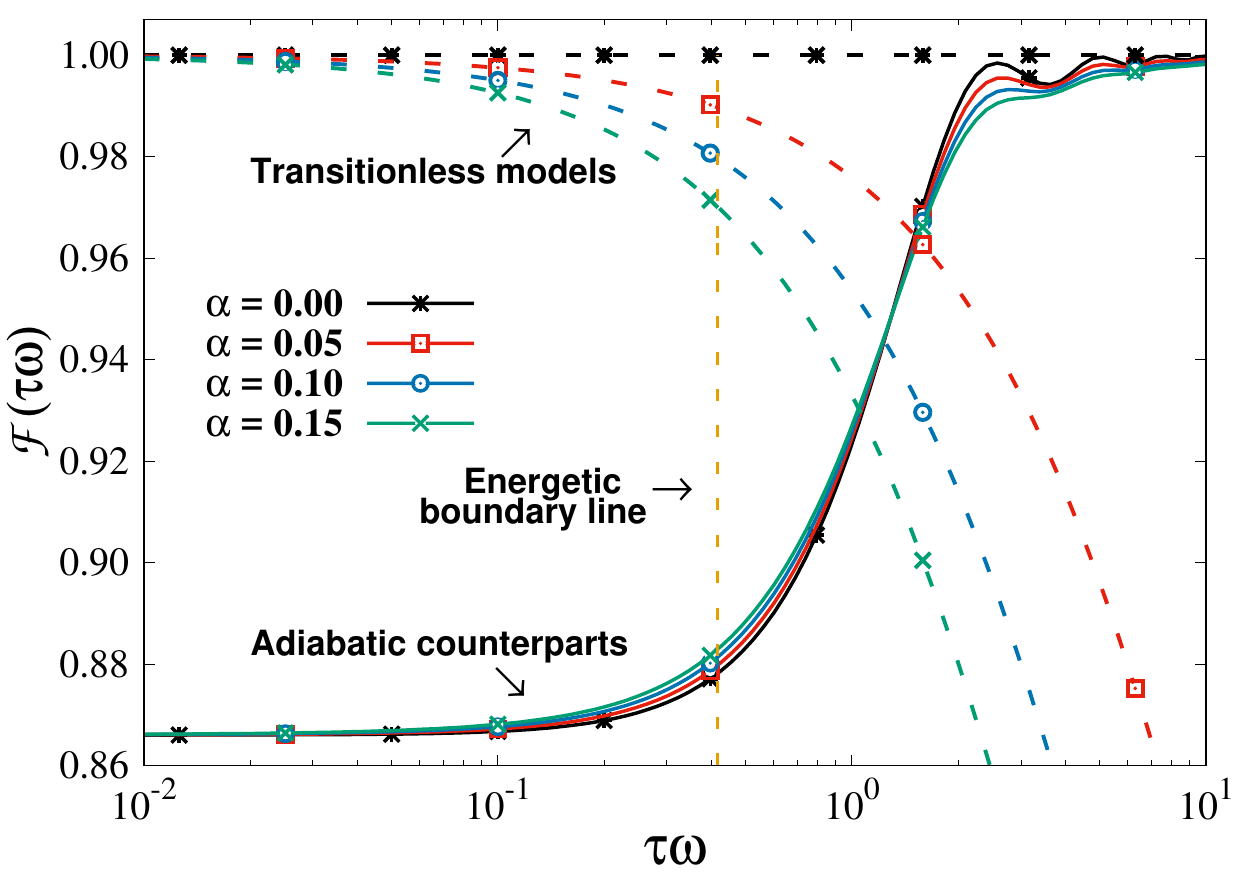}\label{Fig1-d}}
	\caption{(Color online) Fidelity $\cal{F}(\tau)$ under decoherence in the eigenstate basis for both adiabatic (solid curves) and optimal transitionless dynamics (dashed curves) 
		in the Landau-Zener model. Left column: Fidelity $\cal{F}(\tau)$ under GAD for (\ref{Fig1-a}) identical and (\ref{Fig1-c}) different energy resources. 
		Right column: Fidelity $\cal{F}(\tau)$ under dephasing for (\ref{Fig1-b}) identical and (\ref{Fig1-d}) different energy resources. 
		The vertical dashed line in  (\ref{Fig1-c}) and (\ref{Fig1-d}) represents the boundary line between
		$\Sigma _{\text{Ad}}\left( \tau \right) \leqslant {\Sigma}_{\text{SA}}\left( \tau \right)$ and 
		$\Sigma _{\text{Ad}}\left( \tau \right) \geqslant {\Sigma}_{\text{SA}}\left( \tau \right)$. We set $\vartheta_{0} = \pi / 3$.}
	\label{LZ-Transit}
\end{figure*}

%%%%%%%%%%%%%%%%%%%%%%%%%
\subsection{Robustness against decoherence in the Landau-Zener model}
%%%%%%%%%%%%%%%%%%%%%%%%%

We are now ready to compare the behavior under decoherence of both transitionless 
and adiabatic models. The system is prepared in the ground state $|E_{-}\left( 0\right) \rangle =|0\rangle $ 
of the Landau-Zener Hamiltonian $H_0^{\text{LZ}}\left( s\right) $ at $s=0$. Then, we 
let the system evolve aiming at the target state $|E_{-}\left( 1 \right) \rangle$. 
We adopt the fidelity 
$\mathcal{F}\left( \tau \right) =\sqrt{ {\langle E_{-}\left( 1\right) |\rho\left( 1\right) |E_{-}\left( 1\right) \rangle }}$ 
as a success measure of each protocol, with $\rho \left(1\right) $ denoting the solution of 
Eq. (\ref{parLindEq}) at $s=1$. To settle the problem in a fair scenario, we shall submit both models to the same requirements of energy cost and total evolution time $\tau$. The robustness of adiabatic and optimal transitionless 
evolutions under GAD and dephasing, for the same and different energetic resources, are shown in Fig.~\ref{LZ-Transit}. To both situations the decoherence rate strength is controlled by the dimensionless parameter $\alpha$. For equal energy resource provided for adiabatic and 
generalized transitionless evolutions, $\omega_r$ will be taken as follows. We consider a set 
$\{ \tau_{i} \, | \, 1 \le i \le n \}$ of total evolution times. 
The total time $\tau_i$ fixes the energy of the generalized transitionless evolution, with faster evolutions related to shorter times. 
For a given $\tau_i$, we adjust the corresponding frequency $\omega_i$ of the Hamiltonian that drives 
the adiabatic evolution so that $\Sigma_{\text{Ad}} (\tau_i) = \Sigma_{\text{SA}} (\tau_i)$, with $\Sigma_{\text{Ad}}(\tau_i)$ 
denoting the energy cost of the adiabatic model. The relevant frequency $\omega_r$ that sets the 
decoherence rates $\gamma_{\pm}$ will then be defined by the average of $\omega _{i}$  for the values 
of $\tau_i$ considered. More specifically, 
${\omega} _{r} \equiv \frac{1}{n} \sum \nolimits _{i=1}^{n} \omega _{i} $, with $n=200$ in our numerical 
treatment.

By considering the situation of identical energy cost [see Figs.~(\ref{Fig1-a}) and~(\ref{Fig1-b})], the fidelity for unitary 
dynamics ($\alpha = 0$) in the adiabatic model is constant and smaller than one. This is because 
of the requirement of fixed energy given by Eq.~(\ref{constraint-er}), which imposes a fixed relationship 
between $\tau$ and $\omega$. The relation between $\tau$ and $\omega$ keeps the adiabatic condition 
unchanged as we increase $\tau$, since we will have to decrease $\omega$ at the same pace. 
On the other hand, transitionless evolutions have fidelity close to $1$, since they are {\it not} ruled 
by the adiabatic constraint. For non-unitary evolutions ($\alpha > 0$), generalized transitionless 
evolutions are more robust than their adiabatic counterparts for any value of $\alpha$ within a range of values 
for $\tau$. Note also that, for the GAD channel in Fig.~(\ref{Fig1-a}), fidelity decreases for 
intermediate times due to the population of excited states in a thermal environment and then is 
favored for long times due to the spontaneous emission effect in the energy eigenbasis. 
In particular, it approximates to the adiabatic fidelity for closed systems for $\tau \rightarrow \infty$. 
Remarkably, the fidelity of the adiabatic curves increases under 
dephasing in the eigenstate basis, as shown in Fig.~(\ref{Fig1-b}). 
For this specific case, this occurs due to the fact that the ground 
eigenprojection $|E_{-}\left( s\right) \rangle \langle E_{-}\left( s\right) |$ is an eigenstate of the Lindblad 
superoperator, which governs the adiabatic approximation in open quantum 
systems~\cite{Sarandy:05,Sarandy:05-2}. Since adiabaticity is governed by the eigenvalue scale of the 
Lindblad superoperator instead of the Hamiltonian eigenvalue scale, Eq.~(\ref{constraint-er}) does not 
prevent the increase of the adiabatic fidelity as it happens in the closed case. Indeed, decoherence 
enhances adiabaticity in this situation. 

Similar results are also shown in Figs.~(\ref{LZ-Transit}{\color{blue}c}) and (\ref{LZ-Transit}{\color{blue}d}), where 
we allow for different resource contents. For this case, the relevant frequency is simply 
adopted as $\omega_r \equiv \omega$.
Observe that the behavior of the fidelity curve 
on the right and left hand side of the vertical line shows that, even for more 
energy provided for the adiabatic model, the transitionless dynamics can be more robust 
than the adiabatic dynamics for a fixed $\alpha$. Therefore, generalized 
transitionless evolutions can be more robust in a real open-system scenario even in situations 
for which the adiabatic implementation has more energy resource available.
For all situations considered in Fig.~\ref{LZ-Transit}, the crossing points delimit the supremacy region of the optimal transitionless dynamics. This region depends of the coupling strength between the qubit and its reservoir 
(as measured by the parameter $\alpha$). Therefore, in general, the advantage of the optimal 
transitionless evolution is a non-trivial problem, which depends on both the 
decoherence channel and the coupling strength with the reservoir.

We observe that the generalized counter-diabatic theory can be shown to be more 
robust than its adiabatic counterpart for different values of temperature, with the 
the choice $N_{\text{th}}=1/2$ just accounting for a simple numerical instance. In particular, 
the plots for each value of the parameter $\alpha$ in Fig.~\ref{LZ-Transit} already indicate 
that the advantage holds for distinct temperature regimes. 
More specifically, provided the expression for the parameters $\gamma_{+}$  and  $\gamma_{-}$
in terms both of the decoherence rate $\gamma_{0}$ and the temperature (which 
is implicit in $N_{\text{th}}$), we can think of the different values for the 
parameter $\alpha$ either as a change in the decoherence rate $\gamma_{0}$
(keeping $N_{\text{th}}$ fixed) or as a change in the temperature parameter $N_{\text{th}}$ 
(keeping $\gamma_0$ fixed). Therefore, different values of $\alpha$ can be 
taken as yielded by a change in the temperature of the bath.

%%%%%%%%%%%%%%%%%%%%%%%%%%%%%%%%%%%%%%%%%%
\section{Transitionless dynamics in the counter-diabatic gate model}
%%%%%%%%%%%%%%%%%%%%%%%%%%%%%%%%%%%%%%%%%%

Shortcuts to adiabaticity can be used to 
speed up adiabatic quantum gates. More specifically, they have been applied to perform 
universal quantum computation (QC) via either counter-diabatic controlled evolutions~\cite{SciRep:15} or 
counter-diabatic quantum teleportation~\cite{PRA:16}. As hybrid models, these approaches provide a convenient 
digital architecture for physical realizations while potentially keeping both the generality and some inherently 
robustness of analog implementations. Experimentally, digitized implementations of quantum annealing processes  
have been recently provided~\cite{Barends:16}, with controlled quantum gates adiabatically 
realized with high fidelity via superconducting qubits~\cite{Martinis:14}.  
In this Section, by focusing on controlled evolutions, we will now show that 
counter-diabatic QC can be more robust against decoherence than its adiabatic 
counterpart as long as the gate runtime is suitably determined within a range of evolution times.

%%%%%%%%%%%%%%%%%%%%%%%%%%%%%%%%%%%%%%%%%%
\subsection{Adiabatic and counter-diabatic controlled quantum dynamics}
%%%%%%%%%%%%%%%%%%%%%%%%%%%%%%%%%%%%%%%%%%

Consider a bipartite system composed by a target subsystem ${\cal T}$ and an auxiliary subsystem ${\cal A}$, whose 
individual Hilbert spaces ${\cal H}_{\cal T}$ and ${\cal H}_{\cal A}$ have dimensions $d_{\cal T}$ and $d_{\cal A}$, respectively. 
The auxiliary subsystem ${\cal A}$ will be driven by a family of time-dependent Hamiltonians $\{H_{k}\left( s\right)\}$, with $0 \le k \le d_{\cal T}-1$. 
The target subsystem will be evolved by a complete set $\left\{ P_{k}\right\}$ of orthogonal
projectors over ${\cal T}$, which satisfy $P_{k}P_{m}=\delta _{km}P_{k}
$ and $\sum_{k}P_{k}=\1$. In a controlled adiabatic dynamics, 
the composite system ${\cal TA}$ will be governed by a Hamiltonian in the form~\cite{Hen:15} 
\begin{equation}
H\left( s\right) =\sum\nolimits_{k}P_{k}\otimes H_{k}\left( s\right), 
\label{sce.1.2}
\end{equation}
with 
$H_{k}\left( s\right) =g\left( s\right) H_{k}^{\left( f\right)}+f\left( s\right) H^{\left( b\right) }$, 
where $H^{(b)}$ is the beginning Hamiltonian, $H^{(f)}_k$ is the contribution $k$ to the final Hamiltonian, and the time-dependent 
functions $f(s)$ and $g(s)$ satisfy the 
boundary conditions $f\left( 0\right) =g\left( 1\right) =1$ and $g\left( 0\right) =f\left(1\right) =0$.

Suppose now we prepare ${\cal TA}$ in the initial state 
$\left\vert \Psi _{\text{init}}\right\rangle =\left\vert \psi\right\rangle \otimes \left\vert \varepsilon _{b}\right\rangle$, where 
$\left\vert \psi\right\rangle$ is an arbitrary state of ${\cal T}$ and 
$\left\vert \varepsilon _{b}\right\rangle$ is the (non-degenerate) ground state of $H^{\left( b\right) }$. Then $\left\vert \Psi _{\text{init}}\right\rangle$ 
is the ground state of the initial Hamiltonian $\mathbbm{1}\otimes H^{\left( b\right) }$. By applying the adiabatic theorem~\cite{Messiah:book,Sarandy:04} , 
a sufficiently slowing-varying evolution of $H(t)$ will drive the system (up to a phase) to the final state 
\begin{equation}
\left\vert \Psi_{\text{final}}\right\rangle =\sum\nolimits_{k}P_{k}\left\vert \psi \right\rangle \otimes \left\vert \varepsilon _{k}\right\rangle,
\label{psi-final}
\end{equation}
where $\left\vert \varepsilon_{k}\right\rangle$ is the ground state of $H_{k}^{\left( f\right) }$~\cite{Hen:15}. 
Note that an arbitrary projection $P_k$ over the \textit{unknown} state $\left\vert \psi \right\rangle$ can be yielded by performing a convenient measurement over ${\cal A}$. 
In particular, as will be shown in Subsection~\ref{iv.b}, by suitably designing the auxiliary Hamiltonians $H_{k}\left( s\right)$, such a dynamics can be used to 
adiabatically implement individual quantum gates. 

The counter-diabatic version of this controlled evolution has been built in Ref. \cite{SciRep:15}, where it is shown that the transitionless 
Hamiltonian for the composite system ${\cal TA}$ reads
\begin{equation}
H_{\text{SA}}\left( s\right) =\sum\nolimits_{k}P_{k}\otimes H_{\text{SA},k}\left(
s\right) ,  \label{sce1.7}
\end{equation}%
where $H_{\text{SA},k}\left( s\right)$ is the piecewise Hamiltonian implementing the shortcut to adiabaticity for the controlled dynamics.

%%%%%%%%%%%%%%%%%%%%%%%%%%%%%%%%%%%%%%%%%%
\subsection{Quantum computation via adiabatic controlled evolutions}
%%%%%%%%%%%%%%%%%%%%%%%%%%%%%%%%%%%%%%%%%%
\label{iv.b}

Universal sets of quantum gates can be implemented through a bipartite system $\mathcal{TA}$ 
composed by a target subsystem $\mathcal{T}$ and a single-qubit 
auxiliary system $\mathcal{A}$. In our protocol, the target system works as our quantum processor, with any computation performed on it. In others words, both the input and output state, as well as any intermediate stage of the computation, should be encoded in the target system. For example, the target \textit{system} for a two-qubit gate is composed by the target \textit{qubit} and the control \textit{qubit}. On the other hand, the auxiliary qubit works as an \textit{ancilla qubit}. Any result of measurements over such a qubit is not relevant for the computation result, but it is important for determining whether or not the computation has been successfully realized. Differently from the target system, any information encoded in the \textit{auxiliary} system may be deleted after the measure. Therefore, quantum gates will be applied to the target subsystem, 
as a result of a measurement performed on the auxiliary qubit.

Let us begin by considering $\mathcal{T}$ as a single qubit 
and a single-qubit gate as an arbitrary 
rotation of angle $\phi$ around a direction $\hat{n}$ over the Bloch 
sphere. Under this consideration, the Hamiltonian that adiabatically implements such a single-qubit gate for an arbitrary input state 
$| \psi \ket = a |0 \ket + b |1 \ket$, with $a,b \in \mathbb{C}$, is given by \cite{Hen:15}
\begin{eqnarray}
H_{\text{sg}}\left( s\right) &=& P_{+}\otimes H_{0}\left(
s\right) + P_{-} \otimes H_{\phi}\left( s\right) \text{ ,} \label{HAdsg}
\end{eqnarray}%
where $\{ P_{ \pm } \}$ is a complete set of orthogonal projectors over the Hilbert space of 
the target qubit. 
The projectors can be parametrized as $P_{ \pm } = ( \1 \pm \hat{n} \cdot \vec{\sigma} ) / 2$, with 
$\hat{n}$ associated with the direction of the target qubit on the Bloch sphere. In Eq.~(\ref{HAdsg}), each Hamiltonian $H_{\xi}\left(s\right)$ 
($\xi = \{ 0, \phi \}$) acts on $\A$, and is given by~\cite{Hen:15}
\begin{equation}
H_{\xi }\left( s\right) =-\omega \{ \sigma _{z}\cos (\varphi
_{0}s)+\sin (\varphi _{0}s) [ \sigma _{x}\cos \xi +\sigma _{y}\sin \xi %
] \} \text{ , }
\label{Hphy}
\end{equation}
with $\varphi _{0}$ denoting an arbitrary parameter that sets the success probability 
of obtaining the desired state at the end of the evolution. This parameter plays a 
role in the energy performance of counter-diabatic QC, with probabilistic 
counter-diabatic QC ($ \varphi _{0} \neq \pi $) being energetically more favorable than its 
deterministic ($ \varphi _{0} = \pi $) counterpart~\cite{Front:16}.
\begin{figure}[t]
	\centering
	\includegraphics[width=0.25\textwidth]{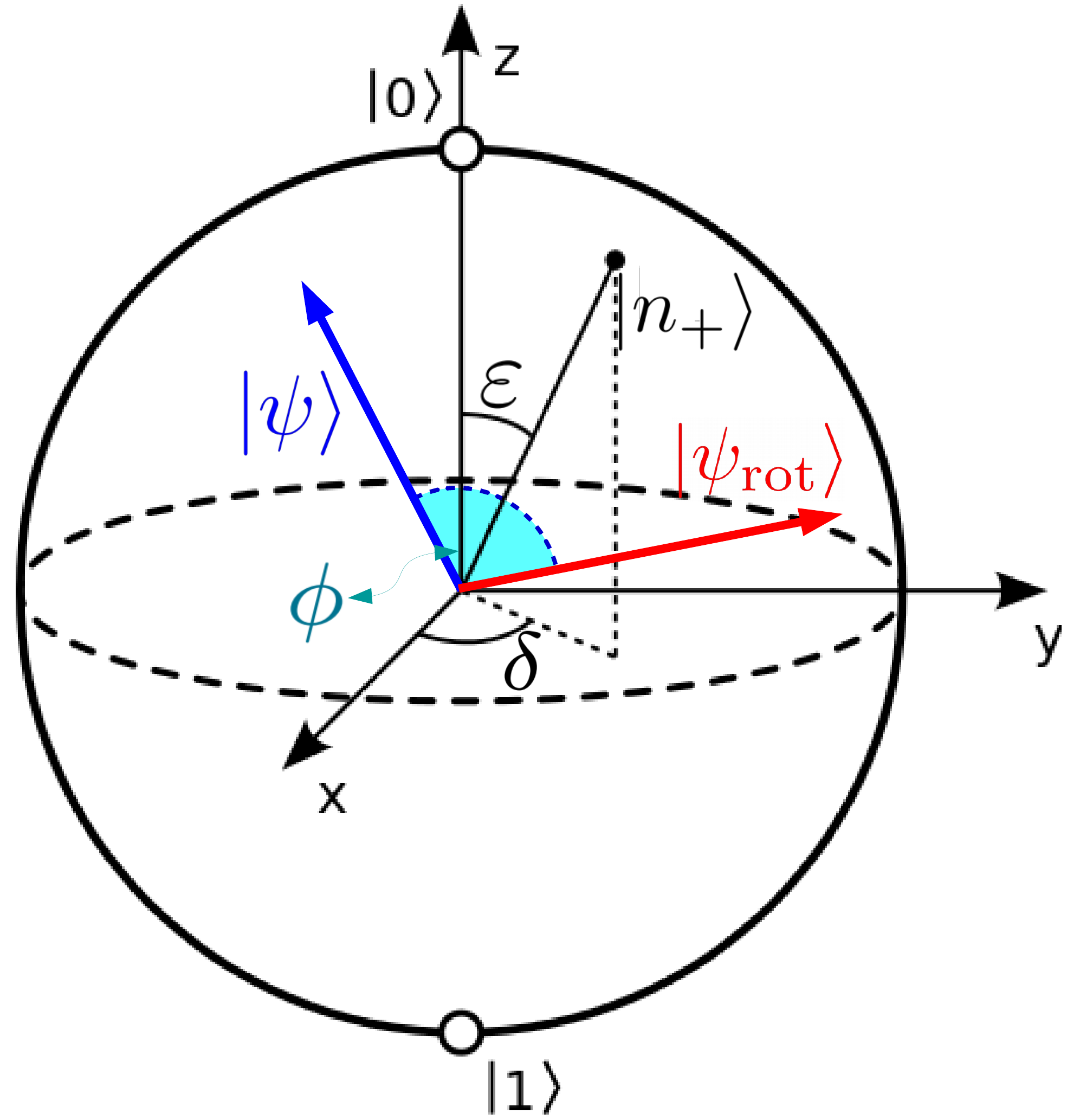}
	\caption{(Color online) Geometric representation of an arbitrary single qubit gate implemented through an adiabatic controlled 
		evolution. Information about the quantum gate to be implemented is encoded in the angles $\varepsilon$ and $\delta$ that set the vector $|n_{+}\ket$ and in the angle $\phi$ that sets the Hamiltonian in Eq. (\ref{Hphy}).}
	\label{GateRep}
\end{figure}
The projectors $\{ P_{ \pm } \}$ may be written in terms of two basis vectors $\{ | n_{\pm} \ket \}$ in the Bloch sphere as $\{ P_{ \pm } \} = | n_{\pm} \ket \bra n_{\pm} |$, where
\begin{eqnarray}
| n_{+} \ket &=& \cos (\varepsilon/2) | 0 \ket + e^{i\delta} \sin (\varepsilon/2) | 1 \ket \label{nmais} \\
| n_{-} \ket &=& -\sin (\varepsilon/2) | 0 \ket + e^{i\delta} \cos (\varepsilon/2) | 1 \ket \label{nmenos} \text{ .}
\end{eqnarray}
Thus, a quantum gate is encoded as a rotation of $\phi$ around the vector $| n_{+} \ket$, as shown in the Fig. \ref{GateRep}. Now, by expressing the state $| \psi \ket$ in the basis $\{ | n_{\pm} \ket \}$, we write 
$| \psi \ket =  \alpha | n_{+} \ket + \beta | n_{-} \ket$, with $\left\vert \hat{n}_{\pm }\right\rangle$ 
being a state in the direction $\hat{n}$ and $\alpha,\beta \in \mathbb{C}$. 
We then prepare the system in the initial state $| \Psi (0) \ket = | \psi \ket | 0 \ket$. Then, 
assuming an adiabatic dynamics, the evolved state $| \Psi (s) \ket $ 
is given by the superposition 
\begin{eqnarray}
| \Psi (s) \ket &=& \alpha | n_{+} \ket | E_{-,0} (s) \ket + \beta | n_{-} \ket | E_{-,\phi} (s) \ket \nonumber \\
&=& \cos \left( \frac{\varphi_{0} s}{2} \right) | \psi \ket | 0 \ket + \sin \left( \frac{\varphi_{0} s}{2} \right) | \psi_{\text{rot}} \ket | 1 \ket \text{ ,} \label{RefereeA1}
\end{eqnarray}
with $| \psi_{\text{rot}} \ket = \alpha | n_{+} \ket + e^{i \phi} \beta | n_{-} \ket$ being the rotated 
desired state and the ground $| E_{-,\xi} (s) \ket$ and first excited $| E_{+,\xi} (s) \ket $ states 
of $H_{\xi }\left( s\right)$ given by
\begin{eqnarray}
| E_{-,\xi} (s) \ket  &=& \cos ( \varphi_{0} s / 2 ) | 0 \ket  + e^{i \xi } \sin ( \varphi_{0} s / 2 ) | 1 \ket \text{ ,} \label{Fund} \\
| E_{+,\xi} (s) \ket  &=& - \sin ( \varphi_{0} s / 2 ) | 0 \ket + e^{i \xi } \cos ( \varphi_{0} s / 2 ) | 1 \ket \text{ .} \label{Exc} 
\end{eqnarray}

We observe that, due to the dynamics of the auxiliary qubit through two adiabatic paths, there are quantum phases 
$\vartheta_{0}(s)$ and $\vartheta_{\phi}(s)$ accompanying the evolutions associated with $| E_{-,0} (s) \ket$ and $| E_{-,\phi} (s) \ket$, respectively. 
Then, relative phases should in principle be considered in Eq. (\ref{RefereeA1}). However, as shown in the Ref. \cite{Hen:15}, such phases satisfy 
$\vartheta_{0}(s) = \vartheta_{\phi}(s)$. Thus, they factorize as a global phase of the state $| \Psi (s) \ket$. 
At the end of the evolution, a measurement on the auxiliary qubit 
yields the rotated state with probability $\sin ^2 \left( \varphi_{0} /2 \right)$ and the input state 
with probability   $\cos ^2 \left( \varphi_{0} /2 \right)$. The computation process is therefore 
\textit{probabilistic}, which succeeds if the auxiliary qubit ends up in the state $|1\ket$. Otherwise, 
the target system automatically returns to the input state and we simply restart the protocol. 
In the adiabatic scenario, the parameter $\varphi_{0}$ can then be adjusted in order to obtain the 
optimal fidelity $1$ by taking the limit $\varphi_{0} \rightarrow \pi$, implying in a deterministic 
computation. 

\begin{figure}[t]
\centering
	{\includegraphics[height=3.5cm]{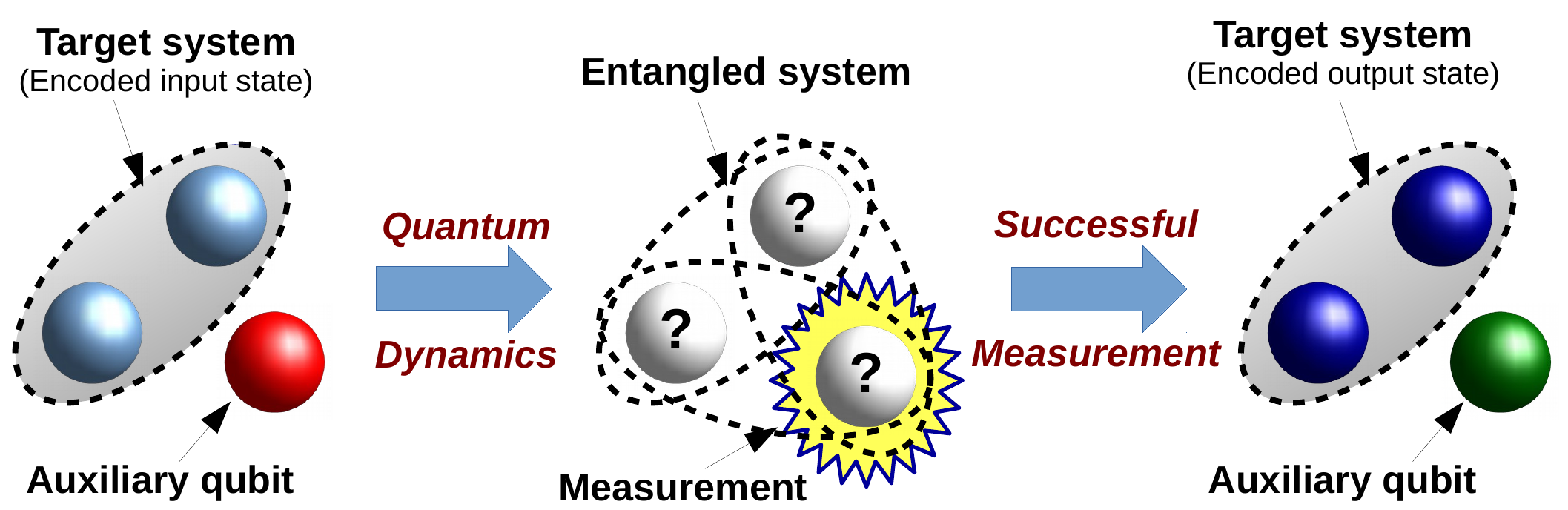} \label{SchemeSI}} 
	\caption{(Color online) Protocol for a probabilistic implementation of a controlled evolution in a two-qubit state. 
		Before the quantum evolution (either adiabatic or nonadiabatic) the input state is encoded in the 
		target system. After the evolution, a measurement (in computational basis) is performed on the 
		auxiliary qubit. A successful measurement corresponds to $| 1 \ket$. If the result is $| 0 \ket$, 
		the system returns to its initial state and a repetition of the process is required (until $| 1 \ket$ 
		is obtained as a result of the measurement).}
	\label{fig-sch}
\end{figure}

This model can be easily adapted to implement controlled single-qubit gates. To this end, the target system 
has to be increased from one qubit to two qubits, as shown in the scheme provided in 
Fig.~\ref{fig-sch}. Here, we adopt that the single-qubit gate 
acts on the target register if the state of the control register is $| 1 \ket$. With this convention, 
the Hamiltonian that implements a controlled single-qubit gate is given by
\begin{eqnarray}
H_{\text{cg}}\left( s\right) &=& (\1 - P_{1,-})\otimes H_{0}\left( s\right) + P_{1,-} \otimes H_{\phi}\left( s\right) \text{ ,} \label{HAdcg}
\end{eqnarray}
where now the set the orthogonal projectors is given by 
$P_{k,\pm} = |k \ket \bra k| \otimes | n_{\pm} \ket \bra n_{\pm} | $, where $|k \ket$ denotes the 
computational basis. The input state of the target system is now written as 
$| \psi_{2} \ket = a | 00 \ket + b | 01 \ket + c | 10 \ket + d | 11 \ket $, with $a,b,c,d \in \mathbb{C}$ and 
$| nm \ket=| n \ket | m \ket$ denoting the control and target register, respectively. 
By rewriting $| \psi_{2} \ket $ in terms of the basis $| n_{\pm} \ket$, we have 
$| \psi_{2} \ket = \alpha | 0 n_{+} \ket + \beta | 0 n_{-} \ket + \gamma | 1 n_{+} \ket + \delta | 1 n_{-} \ket $, 
with $\alpha,\beta,\gamma,\delta \in \mathbb{C}$. 
Therefore, by assuming adiabatic evolution, the system evolves from the state 
$| \Psi_{2}(0) \ket = | \psi_{2} \ket | 0 \ket$, to the instantaneous state
\begin{eqnarray}
| \Psi (s) \ket &=& \alpha | 0 n_{+} \ket | E_{-,0} (s) \ket + \beta | 0 n_{-} \ket | E_{-,0} (s) \ket  \nonumber \\ 
&& + \gamma | 1 n_{+} \ket | E_{-,0} (s) \ket + \delta | 1 n_{-} \ket | E_{-,\phi} (s) \ket \nonumber \\ 
&=& \cos \left( \frac{\varphi _{0} s}{2} \right) | \psi_{2} \ket | 0 \ket + \sin \left( \frac{\varphi_{0} s}{2} \right) | \psi_{2\text{rot}} \ket | 1 \ket 
\end{eqnarray}
with $| \psi_{2\text{rot}} \ket = \alpha | 0 n_{+} \ket + \beta | 0 n_{-} \ket + \gamma | 1 n_{+} \ket + e^{i \phi} \delta | 1 n_{-} \ket $ being the rotated desired state. We then see that the final state $| \Psi (1) \ket $ allows for a probabilistic interpretation for the evolution and, consequently, the computation protocol can again be taken as 
probabilistic ($\varphi_0 \neq \pi$) or deterministic ($\varphi_0 = \pi$).

%%%%%%%%%%%%%%%%%%%%%%%%%%%%%%%
\subsection{Quantum computation via counter-diabatic controlled evolutions}
%%%%%%%%%%%%%%%%%%%%%%%%%%%%%%%

Let us now provide energetically optimal shortcuts to the adiabatic controlled dynamics previously introduced.
The transitionless evolution for the quantum gate Hamiltonian $H_{\text{sg}}\left( s\right)$ is based on the Hamiltonian $H_{\xi }\left( s\right)$ as in  
Eq.~(\ref{Hphy}), such that for single-qubit and controlled single-qubit gates we have~\cite{SciRep:15}
\begin{eqnarray}
H^{\text{SA}}_{\text{sg}} &=& P_{+}\otimes H_{\text{SA},0} + P_{-} \otimes H_{\text{SA},\phi} \text{ ,} \label{HSAdsg}
\\
H^{\text{SA}}_{\text{cg}} &=& (\1 - P_{1,-})\otimes H_{\text{SA},0} + P_{1,-} \otimes H_{\text{SA},\phi} \text{ ,} \label{HSAdcg}
\end{eqnarray}
respectively. Remarkably, each 
Hamiltonian $H_{\xi}\left(s\right)$ satisfies the conditions required by Theorems \ref{TheoOptmEner} 
and \ref{TheoTimeIndep}, so that we can obtain an optimal time-independent  
Hamiltonian.
Thus, the generalized Hamiltonians associated with $H _{\xi} \left( s\right)$ for transitionless dynamics can be directly 
derived from Eq.~(\ref{HSA}). Notice that, from Eq. (\ref{Fund}) and (\ref{Exc}), it is possible show that Theorem \ref{TheoTimeIndep} holds, which implies in
\begin{eqnarray}
H_{\text{SA},\xi} &=& \frac{1}{\tau}\sum_{k=\pm }|d_{s}E_{k,\xi}\left( s\right)
\rangle \langle E_{k}\left( s\right) | \nonumber \\
&=&   \frac{\varphi _{0}}{2\tau}\left[ \sigma _{y}\cos \xi - \sigma
_{x}\sin \xi \right] \label{HSAxi}
\end{eqnarray}%
where we have used that $\bra E_{k,\xi}\left( s\right) | d_{s}E_{k,\xi}\left( s\right) \rangle=0$, with $k \in \{ +,-\}$. The
Hamiltonian in Eq.~(\ref{HSAxi}) improves the gate Hamiltonian derived in Ref.~\cite{SciRep:15}. More specifically, 
Eq.~(\ref{HSAxi}) is energetically optimal and given by a time-independent operator. 

\begin{figure*}[th!]
	\subfloat[ ]{\includegraphics[scale=0.6]{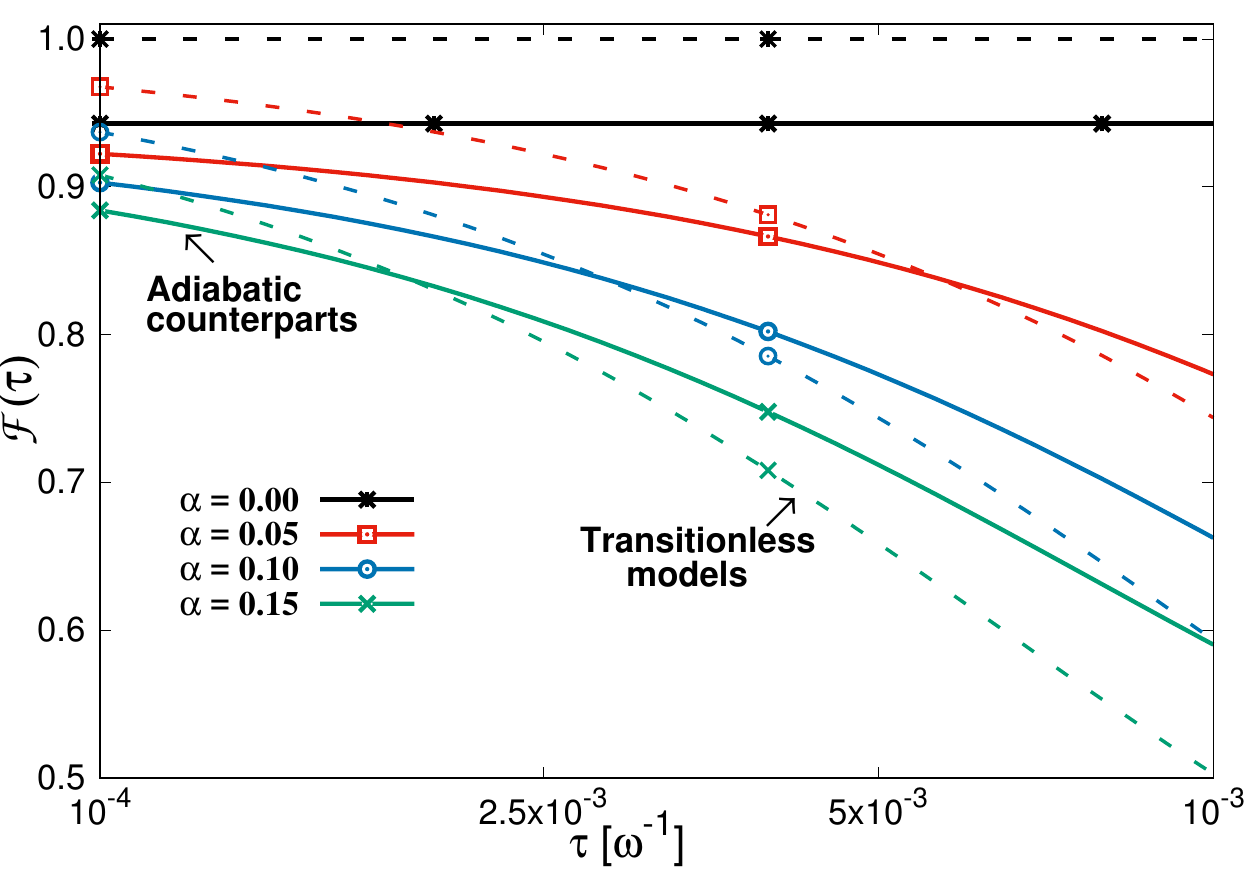}\label{Fig4-a}}\hspace{0.2cm}
	\subfloat[ ]{\includegraphics[scale=0.6]{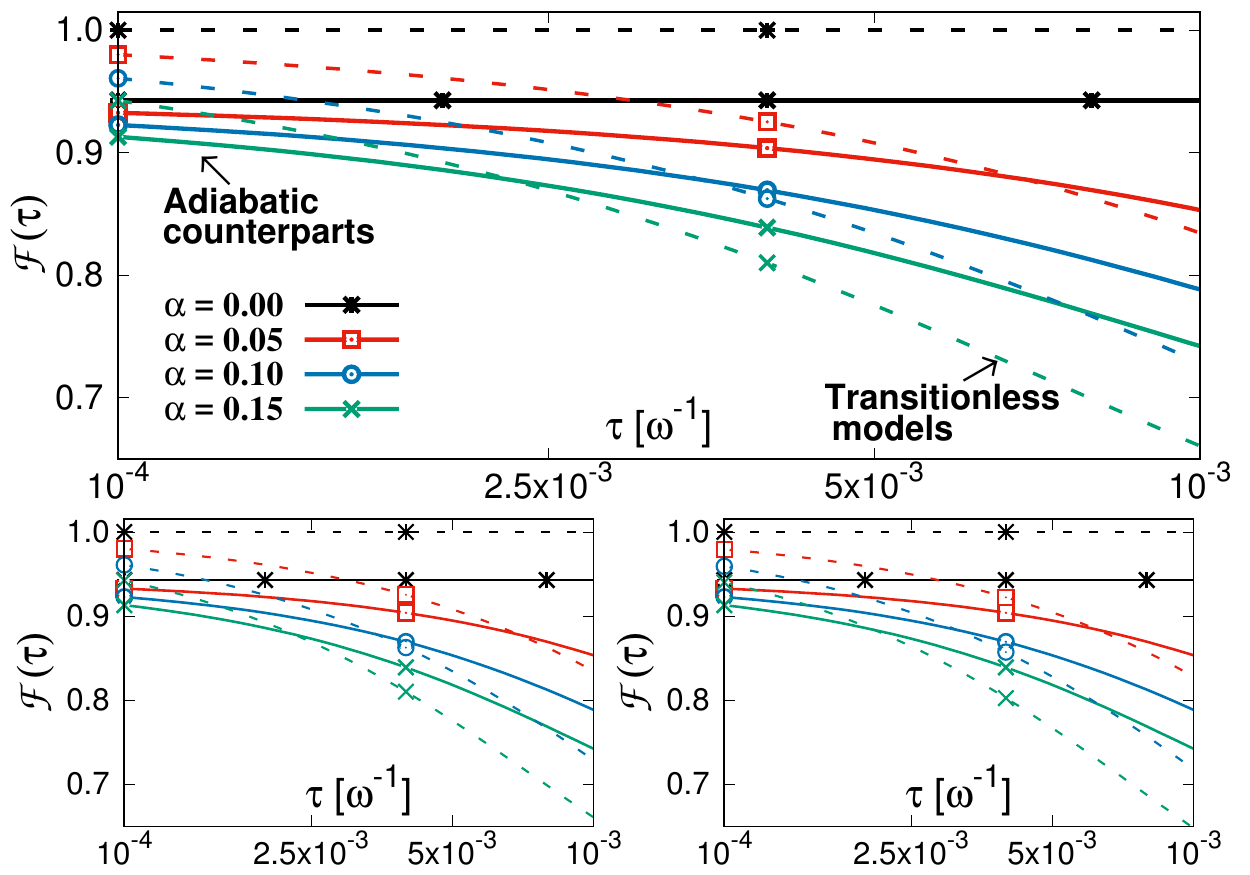}\label{Fig4-b}}
	\caption{(Color online) Fidelity ${\cal F}(\tau)$ for the implementation of (\ref{Fig4-a}) a CNOT gate to the state $| + \ket| 0 \ket$ and (\ref{Fig4-b}) single qubit gates, provided by (\textit{top}) a Hadamard gate to the state $| 0 \ket$, (\textit{bottom left}) a phase gate to the state $| + \ket$, and (\textit{bottom right}) a $\frac{\pi}{8}$-gate to the state $| + \ket$. The gates are implemented via deterministic 
		($\varphi _{0} = \pi$) adiabatic QC (solid curves) and probabilistic ($\varphi _{0} \approx 0.742 \pi$) 
		counter-diabatic QC (dashed curves), for unitary and non-unitary evolutions under dephasing 
		for identical energy resources.  }
	\label{HenProb-ER-Deph}
\end{figure*}
\vspace{0.3cm}
%%%%%%%%%%%%%%%%%%%%%%%%%%%%%%%
\subsection{Energy cost for the quantum gate Hamiltonian}
%%%%%%%%%%%%%%%%%%%%%%%%%%%%%%%
For a transitionless evolution, it is possible to show that a probabilistic 
process, with $\varphi_0 \ne \pi$,  is energetically better than the deterministic approach $\varphi_0 = \pi$~\cite{Front:16}.
For this reason, we will consider here the probabilistic model for the generalized counter-diabatic
quantum gates. 
In this scenario, given a fixed amount of energy resource available, our aim is to compare the best
adiabatic protocol to implement quantum gates with its best generalized transitionless
counterpart. From Eq. (\ref{HSAdsg}), we can write the energy cost of a single evolution to implement single-qubit 
gates 
\begin{equation}
{\Sigma} _{\text{SA,sg}} \left(\tau ,\varphi _{0}\right) = 
\frac{\varphi _{0}}{\omega \tau} \Sigma _{\text{sg}} \text{ ,}
\end{equation}
where $\Sigma _{\text{sg}} = 2 \omega $ corresponds to the adiabatic energy cost $\Sigma _{\text{Ad,sg}}\left( \tau \right)$ and $\varphi _{0}$ is the free angle parameter. The energy cost of probabilistic optimal transitionless evolutions can be obtained by defining the quantity%
\begin{equation}
\langle N \rangle \equiv \frac{1}{\sin ^{2}\left( \varphi _{0}/2\right) } \text{ \ ,}
\label{MediaN}
\end{equation}%
which is the average number of evolutions for a successful computation. Thus, the average energy cost to implement a probabilistic evolution is~\cite{Front:16}
\begin{eqnarray}
\bar{\Sigma}_{\text{SA,sg}}\left(\tau ,\varphi _{0}\right)&=&\langle N \rangle {\Sigma} _{\text{SA,sg}} \left(\tau ,\varphi _{0}\right) \nonumber \\ 
&=& \frac{\varphi _{0}}{\omega \tau} \csc ^2 \left(\varphi _{0} / 2 \right) \Sigma _{\text{sg}} \text{ .}
\label{costsg}
\end{eqnarray}  
Hence, an optimal scheme requires that the choice of $\varphi _{0}$  is such that $\bar{\Sigma}_{\text{SA,sg}}\left(\tau ,\varphi _{0}\right)$ is minimized. In particular, this minimization is obtained for $\varphi _{0} \approx 0.742 \pi$. It is important mention that the energy cost in Eq.~(\ref{costsg}) is obtained by two processes: i) energy minimization through the quantum phase $\theta_n(t)$ that accompanies the transitionless evolution and  ii) application of the probabilistic model of quantum gates. 

By considering the energy rate ${\mathcal{R}}\left(\tau ,\varphi _{0}\right)$ for adiabatic and generalized transitionless protocols [similarly as in Eq.~(\ref{Rfunctiontilde})], we have
\begin{equation}
{\mathcal{R}} \left(\tau ,\varphi _{0}\right) = \frac{\Sigma _{\text{Ad,sg}}\left( \tau \right)}{\bar{\Sigma} _{\text{SA,sg}}\left(\tau ,\varphi _{0}\right)}  = 
\frac{\omega \tau}{\varphi _{0}} \sin ^2 \left( \varphi _{0} / 2 \right) \text{ . }
\label{Rsg}
\end{equation}
By imposing identical energy resource, i.e. ${\mathcal{R}} \left( \tau ,\varphi _{0} \right) = 1$, we obtain
\begin{equation} 
\omega \tau = \varphi _{0} \csc ^2 \left( \varphi _{0} / 2 \right) \text{ .}
\label{omega-tau-cons}
\end{equation}
Remarkably, the energy cost for the implementation of a controlled single-qubit gate by the transitionless 
Hamiltonian in Eq.~(\ref{HSAdcg}) is simply ${\Sigma} _{\text{SA,cg}} = \sqrt{2} {\Sigma} _{\text{SA,sg}}$ \cite{SciRep:15,Front:16}. 
The factor $\sqrt{2}$ also propagates to the adiabatic model, which implies exactly in the same ratio 
${\mathcal{R}} \left(\tau ,\varphi _{0}\right)$ and therefore in the same constraint over $\omega \tau$ provided by 
Eq.~(\ref{omega-tau-cons}). 

%%%%%%%%%%%%%%%%%%%%%%%%%%%%%%%
\subsection{Robustness against decoherence in the quantum gate Hamiltonian}
%%%%%%%%%%%%%%%%%%%%%%%%%%%%%%%

From Eq.~(\ref{Rsg}), the optimal transitionless quantum gate model will be more efficient from the 
energy point of view than its adiabatic counterpart for 
$\omega \tau \geqslant \varphi _{0} \csc ^2 \left( \varphi _{0} / 2 \right)$. Even though this condition 
solves the problem for closed quantum systems, it is a nontrivial problem the energy efficiency of 
the generalized transitionless quantum gate Hamiltonian in comparison with its adiabatic version 
when decoherence effects are not negligible in the physical system. In this section we will 
study the robustness of counter-diabatic QC by controlled evolutions against 
decoherence by considering the case of single-qubit gates in addition to the CNOT gate, which constitute a universal set of quantum gates \cite{Barenco:95}. The success of 
the protocol is measured by the fidelity  $\mathcal{F}(\tau) = \sqrt{ \bra \psi_{\text{rot}} | \rho (1) | \psi_{\text{rot}} \ket  }$, with 
$\rho (1)$ denoting the density operator for the target subsystem, obtained from Eq.~(\ref{parLindEq}).

\begin{figure*}[th!]
	\centering
	\subfloat[ ]{\includegraphics[scale=0.6]{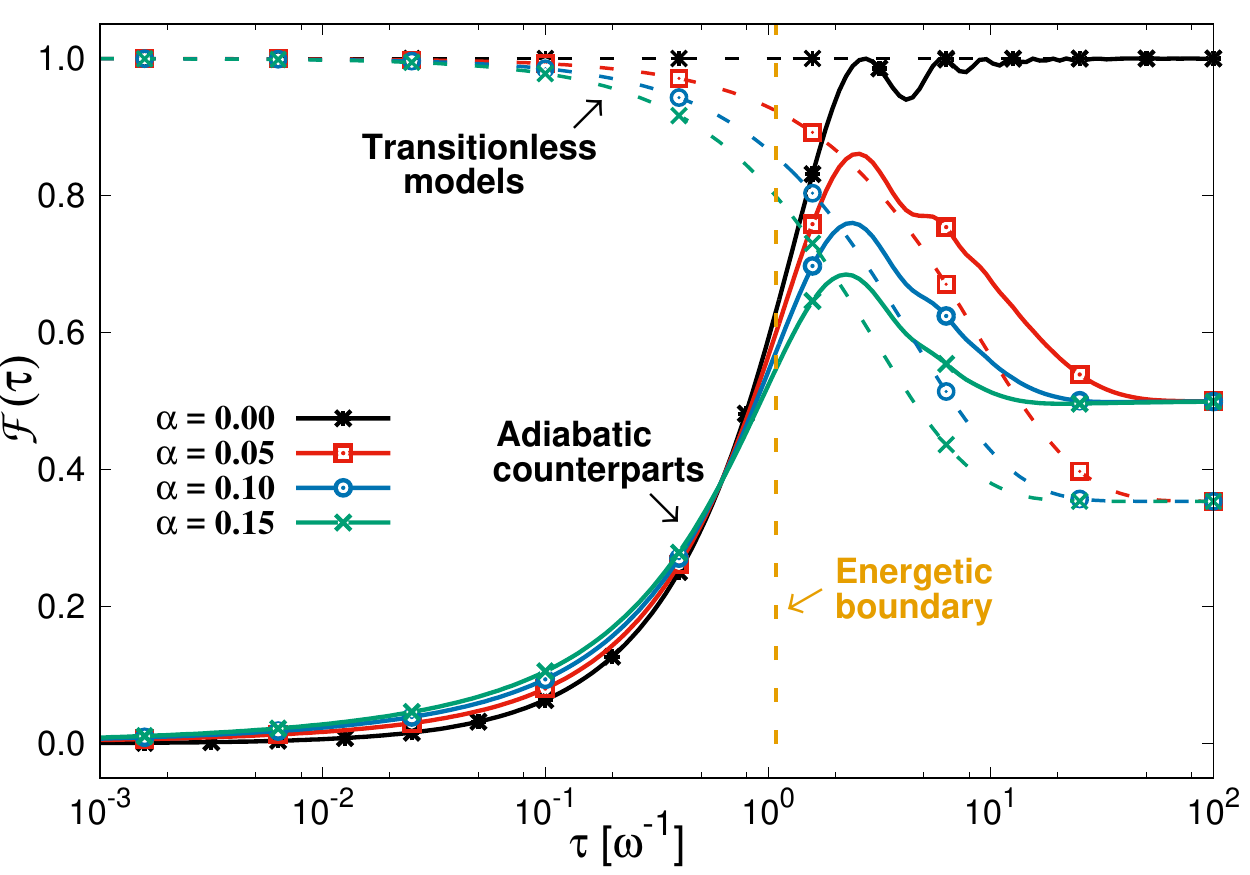}\label{Fig5-a}}\hspace{0.2cm}
	\subfloat[ ]{\includegraphics[scale=0.6]{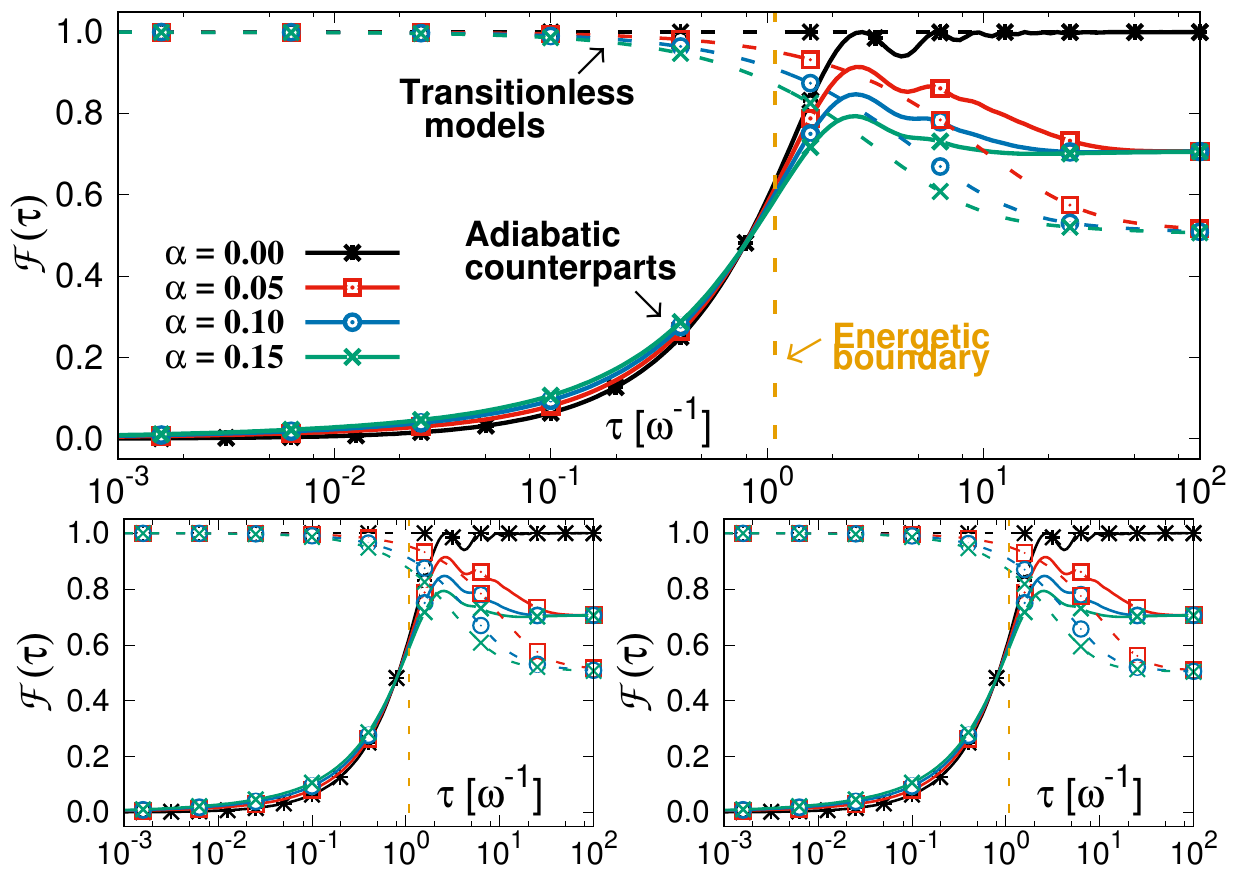}\label{Fig5-b}}
	\caption{(Color online) Fidelity ${\cal F}(\tau)$ for the implementation of (\ref{Fig5-a}) a CNOT gate to the state $| + \ket| 0 \ket$ and (\ref{Fig5-b}) single qubit gates, provided by (\textit{top}) a Hadamard gate to the state $| 0 \ket$, (\textit{bottom left}) a phase gate to the state $| + \ket$, and (\textit{bottom right}) a $\frac{\pi}{8}$-gate to the state $| + \ket$.The gates are implemented via deterministic 
		($\varphi _{0} = \pi$) adiabatic QC (solid curves) and probabilistic ($\varphi _{0} \approx 0.742 \pi$) 
		counter-diabatic QC (dashed curves), for unitary and non-unitary evolutions under dephasing 
		for different energy resources.}
	\label{HenProb-DR-Deph}
\end{figure*}

\begin{figure*}[th!]
	\subfloat[ ]{\includegraphics[scale=0.6]{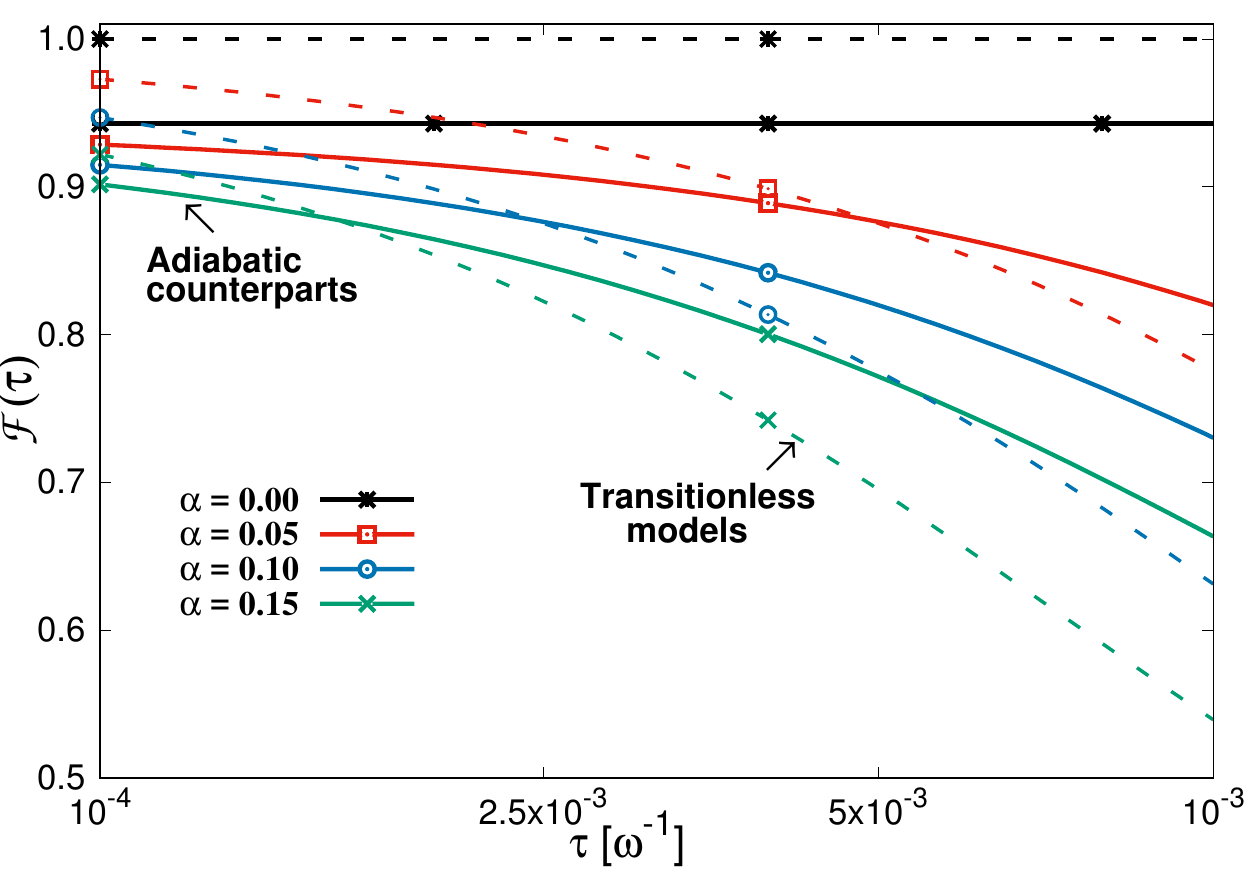}\label{Fig6-a}}\hspace{0.2cm}
	\subfloat[ ]{\includegraphics[scale=0.6]{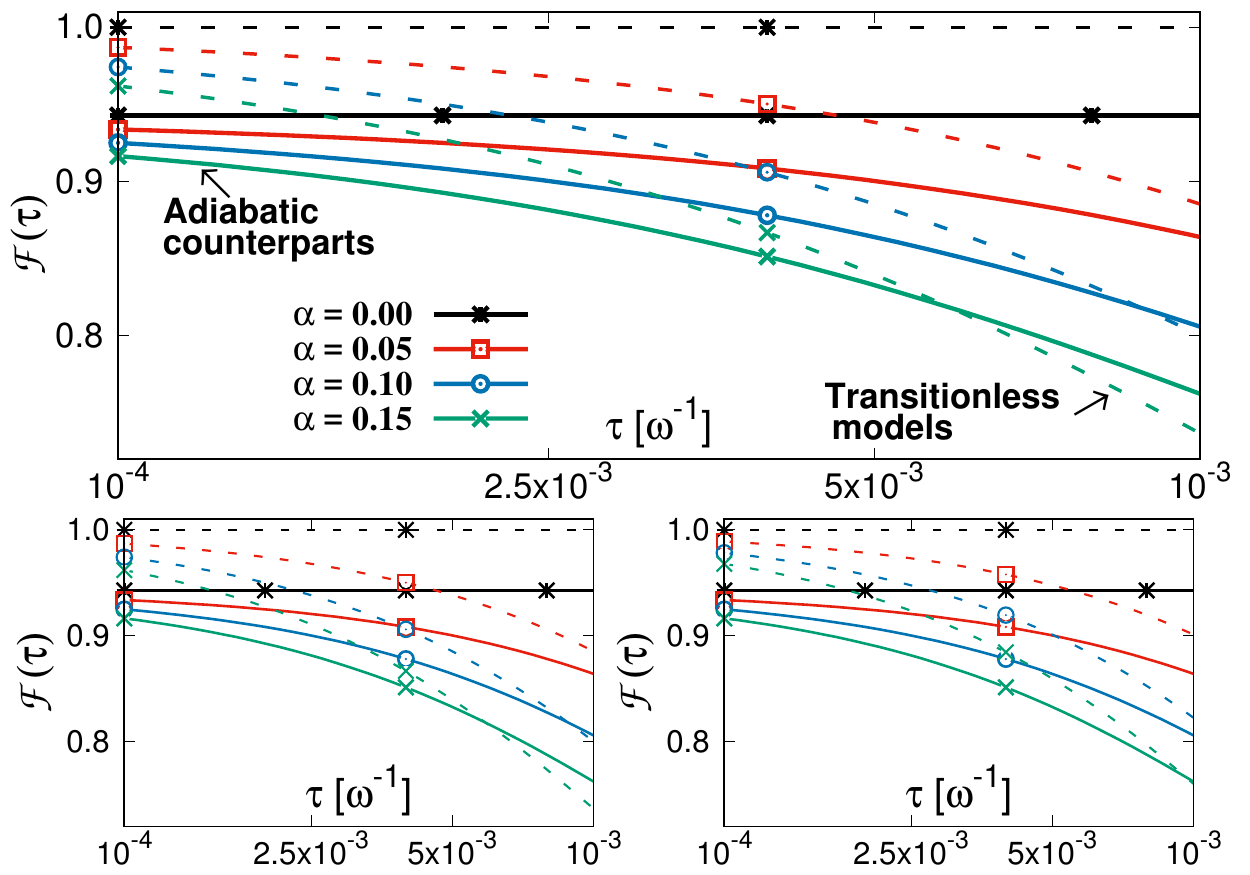}\label{Fig6-b}}
	\caption{(Color online) Fidelity ${\cal F}(\tau)$ for the implementation of (\ref{Fig6-a}) a CNOT gate to the state $| + \ket| 0 \ket$ and (\ref{Fig6-b}) single qubit gates, provided by (\textit{top}) a Hadamard gate to the state $| 0 \ket$, (\textit{bottom left}) a phase gate to the state $| + \ket$, and (\textit{bottom right}) a $\frac{\pi}{8}$-gate to the state $| + \ket$. The gates are implemented via deterministic 
		($\varphi _{0} = \pi$) adiabatic QC (solid curves) and probabilistic ($\varphi _{0} \approx 0.742 \pi$) 
		counter-diabatic QC (dashed curves), for unitary and non-unitary evolutions under GAD 
		for identical energy resources.}
	\label{HenProb-ER-GAD}
\end{figure*}

The Hadamard gate is a rotation of $\pi / 2$ around direction $y$ in the Bloch sphere, so that we set $\phi _{\text{Had}} = \pi / 2$ and $\varepsilon  _{\text{Had}}= \delta  _{\text{Had}}= \pi / 2$ in Eqs. (\ref{nmais}) and (\ref{nmenos}). For the case of phase and $\frac{\pi}{8}$ gates, we take them as rotations around the direction $z$ of an angle $\pi$ and $\pi / 4$, respectively, so that we take $\varepsilon _{\text{pha}}= \varepsilon _{\frac{\pi}{8}}= 0$, $\phi _{\text{pha}} = \pi$, and $\phi _{\frac{\pi}{8}} = \pi / 4$. Concerning the CNOT gate, it is a controlled implementation of the operation $\sigma_{x}$ (flip gate), which can be viewed as a rotation of $\pi$ around the $x$ direction. Thus we set $\varepsilon _{\text{CNOT}} =\pi / 2 $, $\delta _{\text{CNOT}}= 0$, and $\phi_{\text{CNOT}}= \pi$. Therefore, from Eq. (\ref{HSAxi}), we have the following counter-diabatic Hamiltonians
\begin{eqnarray}
H_{\text{CD,0}} &=& \hbar \frac{\varphi _{0}}{2\tau} \sigma _{y} \text{ \ , \ } H_{\text{CD},\frac{\pi}{2}} = - \hbar \frac{\varphi _{0}}{2\tau} \sigma _{x} \text{ , } 
\\
H_{\text{CD,0}} &=& \hbar \frac{\varphi _{0}}{2\tau} \sigma _{y} \text{ \ , \ } H_{\text{CD},\pi} = - \hbar \frac{\varphi _{0}}{2\tau} \sigma _{y} \text{ , }
\\
H_{\text{CD,0}} &=& \hbar \frac{\varphi _{0}}{2\tau} \sigma _{y} \text{ \ , \ } H_{\text{CD},\frac{\pi}{4}} = \hbar \frac{\varphi _{0}}{2 \sqrt{2}  \tau}\left( \sigma _{y} - \sigma
_{x} \right) \text{ , }
\\
H_{\text{CD,0}} &=& \hbar \frac{\varphi _{0}}{2\tau} \sigma _{y} \text{ \ , \ } H_{\text{CD},\pi} = - \hbar \frac{\varphi _{0}}{2\tau} \sigma _{y} \text{ , }
\end{eqnarray}
for Hadamard, phase, $\frac{\pi}{8}$, and CNOT gates, respectively.
In order to study the robustness of single qubit gates we have considered the input stats $| \psi_{\text{Had}} \ket = | 0 \ket$ for Hadamard operation and $| \psi _{\text{pha}}\ket = | + \ket = (1/\sqrt{2})(| 0 \ket + | 1 \ket)$ for phase and $\frac{\pi}{8}$-gate. On the other hand, for the CNOT gate, we consider the initial state $| \psi (0) \ket = | + \ket| 0 \ket$ and apply the gate 
Hamiltonian to create a Bell state $|\psi_{00}\rangle = (1/\sqrt{2}) (|00\rangle + |11\rangle)$. 
Fidelity is then obtained from the explicit solution of the Lindblad equation for $\rho(1)$. For instance, for CNOT, we have $\mathcal{F}\left( \tau \right) =\sqrt{\langle \psi_{00} |\rho\left( 1\right) |\psi_{00} \rangle }$. 

\begin{figure*}[th!]
	\centering
	\subfloat[ ]{\includegraphics[scale=0.6]{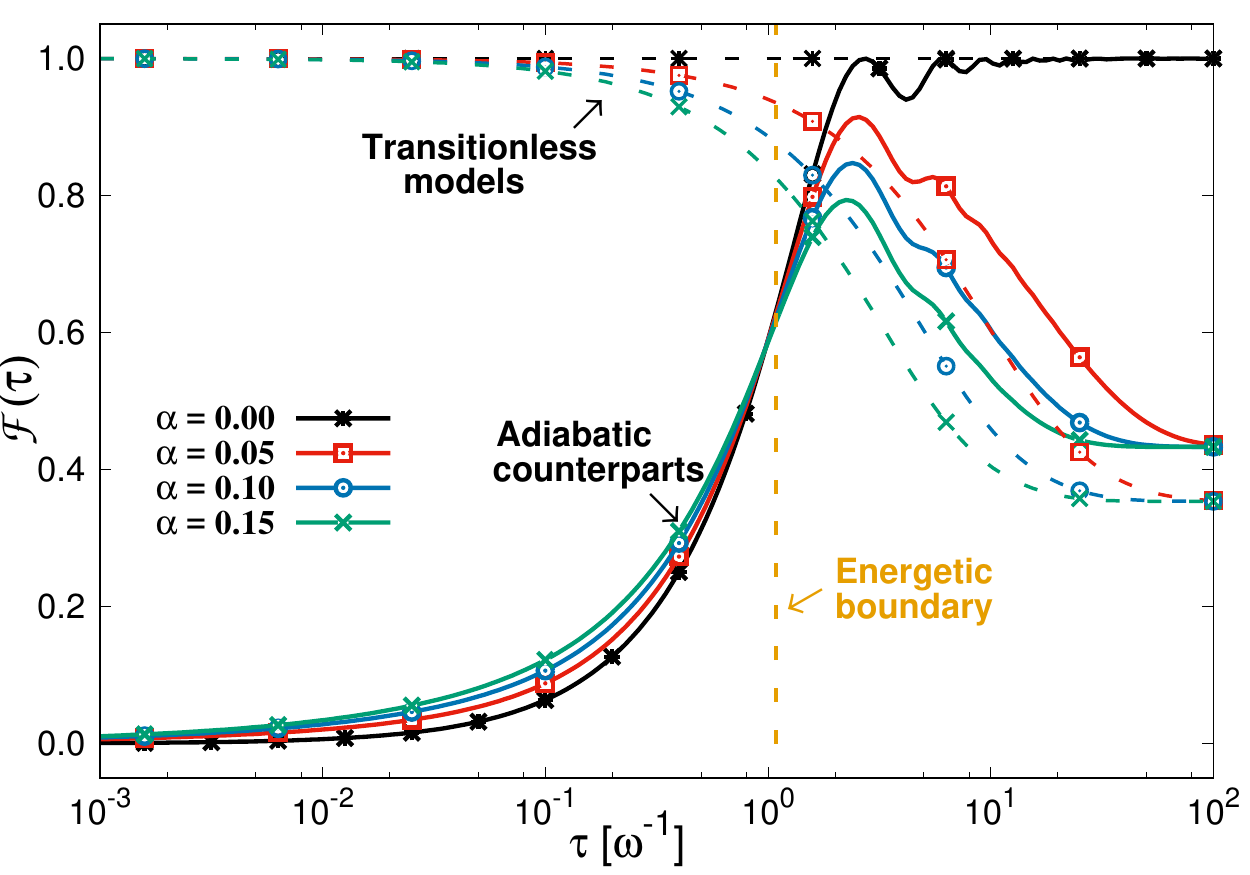}\label{Fig7-a}}\hspace{0.2cm}
	\subfloat[ ]{\includegraphics[scale=0.6]{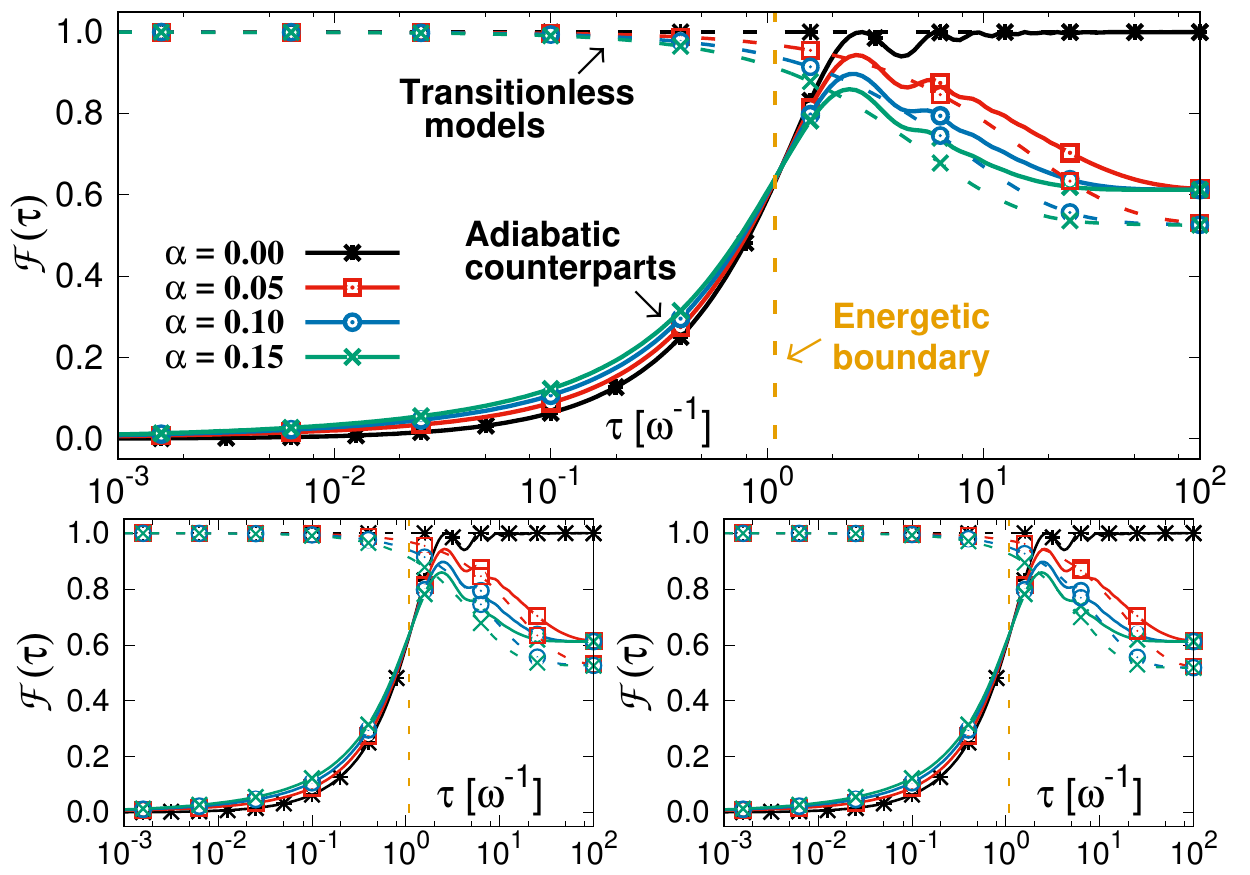}\label{Fig7-b}}
	\caption{(Color online) Fidelity ${\cal F}(\tau)$ for the implementation of (\ref{Fig7-a}) a CNOT gate to the state $| + \ket| 0 \ket$ and (\ref{Fig7-b}) single qubit gates, provided by (\textit{top}) a Hadamard gate to the state $| 0 \ket$, (\textit{bottom left}) a phase gate to the state $| + \ket$, and (\textit{bottom right}) a $\frac{\pi}{8}$-gate to the state $| + \ket$. The gates are implemented via deterministic 
		($\varphi _{0} = \pi$) adiabatic QC (solid curves) and probabilistic ($\varphi _{0} \approx 0.742 \pi$) 
		counter-diabatic QC (dashed curves), for unitary and non-unitary evolutions under GAD 
		for different energy resources.}
	\label{HenProb-DR-GAD}
\end{figure*}

The robustness of the universal set of quantum gate under dephasing is illustrated in Figs. \ref{HenProb-ER-Deph} and \ref{HenProb-DR-Deph} for identical and different resources imposed, respectively.
In these plots, we compare the optimal adiabatic (deterministic) implementation with 
its optimal transitionless version (probabilistic computation and optimal quantum phases). Note that 
the generalized transitionless approach shows a higher fidelity for fast dynamics, but there are regimes 
for which the adiabatic approach shows a better fidelity for a fixed $\alpha$.  
Energy optimization in the optimal transitionless model is achieved for $\varphi _{0} \approx 0.742 \pi$~\cite{Front:16}. Figs. \ref{HenProb-ER-GAD} and \ref{HenProb-DR-GAD} show similar results result for non-unitary evolution under GAD, with equivalent and different resources provided to the adiabatic and optimal transitionless model, respectively. 
In any case of decoherence and energetic resource, there always exist dynamical regimes for which the optimal transitionless evolutions are more robust and therefore a preferred approach in a decohering physical environment.

%%%%%%%%%%%%%%%%%%%%%%%%%%%%%%%%%%%%%%%%%%%%

\section{Conclusion}

In summary, we have developed a generalized minimal energy 
demanding counter-diabatic theory, which is able to yield efficient shortcuts to adiabaticity 
via fast transitionless evolutions. Moreover, we have investigated the robustness of adiabatic 
and counter-diabatic dynamics under decoherence by introducing the requirement of fixed
energy resources, so that a comparison is settled down in a fair scenario. Then, we have shown 
both for the Landau-Zener model and for quantum gate Hamiltonians that there always exist 
dynamical regimes for which generalized transitionless evolutions are more robust and 
therefore a preferred approach in a decohering setup. 
This has been shown both for the dephasing and GAD channels acting on the eigenstate bases. 
It is also possible to show the advantage in other bases, such as the computational basis. 
The general picture is that the gain will typically occur during some finite time range, 
disappearing in the limit of long evolution times.
These results are encouraging for 
the generalized transitionless approach in the open-system realm as long as local Hamiltonians 
are possible to be designed. In the specific case of quantum gate Hamiltonians, this approach 
can be applied, e.g. to derive robust local building blocks for analog implementations of quantum 
circuits (see, e.g., Refs.~\cite{Martinis:14,Barends:16}). Experimental realizations, extensions 
for dealing with systematic errors, and generalized 
shortcuts via reservoir engineering are further directions left for future research.

%%%%%%%%%%%%%%%%%%%%%%%%%%%%%%%%%%%%%%%%%%%%
\section*{Acknowledgments}
	We acknowledge Gonzalo Muga and Tameem Albash for useful discussions. 
	A.C.S. is supported by CNPq-Brazil. M.S.S. acknowledges support from CNPq-Brazil 
	(No. 303070/2016-1), FAPERJ (No. 203036/2016), and the Brazilian National Institute 
	for Science and Technology of Quantum Information (INCT-IQ).

%\begin{widetext}
%\appendix
%%%%%%%%%%%%%%%%%%%%%%%%%
%\section{Supplementary Information} \label{SupInf}
%%%%%%%%%%%%%%%%%%%%%%%%%

%%%%%%%%%%%%%%%%%%%%%%%%%
\appendix
%%%%%%%%%%%%%%%%%%%%%%%%%

\section{Proof of Theorem \ref{TheoOptmEner}} \label{Theorem1}
%%%%%%%%%%%%%%%%%%%%%%%%%

\setcounter{theorem}{0}

\begin{theorem} 
	Consider a closed quantum system under adiabatic evolution governed by a Hamiltonian 
	$H_0\left( t\right) $. The energy cost to implement its generalized transitionless counterpart,  
	driven by the Hamiltonian $H_{\text{SA}}(t)$, can be minimized by setting  
	$\theta_{n}\left( t\right) =\theta_{n}^{\text{min}}\left( t\right)=-i\langle \dot{n}_{t}|n_{t}\rangle$.
\end{theorem}

\begin{proof}
	We adopt as a measure of energy cost the Hamiltonian Hilbert-Schmidt norm, which reads
	\begin{equation}
	\Sigma_{\text{SA}} \left( \tau \right) =\frac{1}{\tau }\int_{0}^{\tau }\sqrt{\text{Tr}%
		\left[ {H_\text{SA}}^{2}\left( t\right) \right] }\text{ }dt \text{ ,} \label{CostSI}
	\end{equation}%
	Then, we obtain
	\begin{eqnarray}
	H_{\text{SA}}^{2}\left( t\right) = \hslash ^{2}\sum\nolimits_{n} \left[ \frac{}{}|\dot{n}_{t}\rangle \langle \dot{n}_{t}|+\theta
	_{n}^{2}\left( t\right) |n_{t}\rangle \langle n_{t}|+\right.  \nonumber \\ 
	\hspace{-0.3cm} \left. i\theta _{n}\left(
	t\right) \left( |n_{t}\rangle \langle \dot{n}_{t}|-|\dot{n}_{t}\rangle
	\langle n_{t}|\right]) \frac{}{}\right] \text{ .}
	\label{HSA2}
	\end{eqnarray}
	By taking the trace of $H_{\text{SA}}^{2}\left( t\right) $ in Eq.~(\ref{HSA2}), we have%
	\begin{eqnarray}
	\text{Tr} \left[ H_{\text{SA}}^{2}\left( t\right) \right] =\sum\nolimits_{m}\langle m_{t}|H_{\text{%
			SA}}^{2}\left( t\right) |m_{t}\rangle \hspace{1.7cm} \nonumber \\
	=\hslash ^{2}\sum_{n} \left[\frac{}{} \langle \dot{n}%
	_{t}|\dot{n}_{t}\rangle +\theta _{n}^{2}\left( t\right) +2i\theta _{n}\left(
	t\right) \langle \dot{n}_{t}|n_{t}\rangle \frac{}{}\right].  \label{TrH2}
	\end{eqnarray}
	Then
	\begin{equation}
	\Sigma_{ \text{\text{SA}} } \left( \tau \right) =\frac{1}{\tau }\int_{0}^{\tau }\sqrt{\sum\nolimits\nolimits_{n} \langle \dot{n}_t|\dot{n}_t\rangle + \Gamma_{n}(\theta_{n})}\text{ }dt \text{ \ ,} \label{Supercost}
	\end{equation}
	where we have $\Gamma_{n}(\theta_{n})=\theta_{n}^{2}\left( t\right) +2i\theta _{n}\left( t\right) \langle \dot{n}_t|n_t \ket$.
	
	We can now find out the functions $\theta _{n} \left( t\right)$
	that minimize the energy cost in transitionless evolutions. For this
	end, we minimize the quantity $\Sigma _{\text{SA}}\left( \tau \right) $
	for the Hamiltonian $H_{\text{SA}}\left( t\right) $ with
	respect to parameters $\theta _{n}\left( t\right) $, where we will adopt it
	being independents. By evaluating $\partial _{\theta _{n}}\Sigma \left( \tau
	\right) $, we obtain%
\begin{equation}
	\partial _{\theta _{n}}\Sigma _{\text{SA}}\left( \tau \right) =\frac{1}{%
		2\tau }\int_{0}^{\tau }\frac{\partial _{\theta _{n}}\{\text{Tr}[H_{\text{SA}%
		}^{2}\left( t\right) ]\}}{\sqrt{\text{Tr}\left[ {H}^{2}\left( t\right) %
		\right] }}\text{ }dt \text{ .} 
\end{equation}
We then impose $\partial _{\theta _{n}}\{$Tr$[H_{%
	\text{SA}}^{2}\left( t\right) ]\}=0$ for all time $t\in \left[ 0,\tau \right]$, 
which ensures $\partial _{\theta _{n}}\Sigma _{\text{SA}}\left( \tau \right)=0$. 
Thus, by using Eq. (\ref{TrH2}), we write%
\begin{equation}
\partial _{\theta _{n}}\{\mathrm{Tr}[H_{\text{SA}}^{2}\left( t\right) ]\}=2\theta
_{n}\left( t\right) +2i\langle \dot{n}_{t}|n_{t}\rangle = 0
\text{ .}
\end{equation}
This implies
\begin{equation}
\theta _{n}\left( t\right) = \theta ^{\text{min}}_{n}\left( t\right) = -i\langle \dot{n}_{t}|n_{t}\rangle \text{ .}
\label{OptTetaAp}
\end{equation}
From the second derivative analysis, it follows that the choice for $\theta _{n}\left( t\right)$ as in Eq.~(\ref{OptTetaAp})  
necessarily minimizes the energy cost, namely, $\partial ^{2} _{\theta _{n}} \Sigma _{\text{SA}}\left( \tau \right) \vert _{\theta _{n} = \theta ^{\text{min}}_{n}}>0$, which concludes the proof. 

\end{proof}

%%%%%%%%%%%%%%%%%%%%%%%%%
\section{Proof of Theorem \ref{TheoTimeIndep}} \label{Theorem2}
%%%%%%%%%%%%%%%%%%%%%%%%%
\begin{theorem} 
	Let $H_0\left( t\right) $ be a discrete quantum Hamiltonian, with $\{|m_{t}\rangle\}$ denoting its 
	set of instantaneous eigenstates. If $\{|m_{t}\rangle\}$ satisfies $\langle k_t |\dot{m}_{t}\rangle = c_{km}$, 
	with $c_{km}$ complex constants $\forall k,m$, then a family of time-independent 
	Hamiltonians $H^{\{\theta\}}$ for generalized transitionless evolutions can be defined by setting 
	$\theta_{m}\left( t\right) = \theta$, with $\theta$ a single arbitrary real constant $\forall m$.
\end{theorem}

\begin{proof}
	By taking the time derivative of the Hamiltonian $H_{\text{SA}}(t)$, we obtain
	\begin{equation}
	\dot{H}_{\text{SA}}\left( t\right) =i  \sum\nolimits_{n}\frac{d}{dt}\left[ \frac{%
		{}}{{}}|\dot{n}_{t}\rangle \langle n_{t}|+i\theta _{n}\left( t\right)
	|n_{t}\rangle \langle n_{t}|\frac{{}}{{}}\right] \text{ .}
	\end{equation}
	Then, the matrix elements of $\dot{H}_{\text{SA}}\left( t\right) $ 
	in the eigenbasis $\left\{ |m_{t}\rangle \right\} $ of the Hamiltonian $H_0(t)$ read
	\begin{eqnarray*}
		\langle k_{t}|\dot{H}_{\text{SA}}\left( t\right) |m_{t}\rangle  
		&=&i  \langle k_{t}|\ddot{m}_{t}\rangle +i  \sum\nolimits_{n}\langle k_{t}|\dot{n}_{t}\rangle
		\langle \dot{n}_{t}|m_{t}\rangle  \nonumber \\
		&& \hspace{-2cm}-  \left[ \dot{\theta}%
		_{k}\left( t\right) \delta _{km} 
		+\theta _{m}\left( t\right) \langle k_{t}|%
		\dot{m}_{t}\rangle +\theta _{k}\left( t\right) \langle \dot{k}%
		_{t}|m_{t}\rangle \right] .  
	\end{eqnarray*}
	Now, by using $\langle k_{t}|\dot{n}_{t}\rangle =-\langle \dot{k}_{t}|{n}%
	_{t}\rangle $, we write $\langle k_{t}|\dot{n}_{t}\rangle \langle \dot{n}%
	_{t}|m_{t}\rangle =\langle \dot{k}_{t}|n_{t}\rangle \langle n_{t}|\dot{m}%
	_{t}\rangle $ and thus%
	\begin{eqnarray}
	\langle k_{t}|\dot{H}_{\text{SA}}\left( t\right) |m_{t}\rangle 
	&=&i  \frac{d}{d{t}}\left[ \langle k_{t}|\dot{m}_{t}\rangle \right]  \nonumber \\
	&& \hspace{-2.5cm}- 
	\left\{ \dot{\theta}_{k}\left( t\right) \delta _{km}+\left[ \theta
	_{m}\left( t\right) -\theta _{k}\left( t\right) \right] \langle k_{t}|\dot{m}%
	_{t}\rangle \right\} \text{ .} \label{finaleEq}
	\end{eqnarray}%
	For $k=m$ in Eq.~(\ref{finaleEq}), 
	we impose the vanishing of the diagonal elements of $\dot{H}_{\text{SA}}(t)$, namely, 
	$\langle k_{t}|\dot{H}_{\text{SA}}\left(t\right) |k_{t}\rangle =0$. This yields
	\begin{equation}
	\dot{\theta}_{m}\left( t\right) =i\frac{d}{dt}\left[ \langle m_{t}|\dot{m}%
	_{t}\rangle \right] \text{ ,} \label{DiagSI}
	\end{equation}%
	On the other hand, for $k \neq m$ in Eq. (\ref{finaleEq}), we now impose the vanishing 
	of the off-diagonal elements of $\dot{H}_{\text{SA}}(t)$, 
	namely, $\langle k_{t}|\dot{H}_{\text{SA}}\left(
	t\right) |m_{t}\rangle =0$ $(k \ne m)$. This yields 
	\begin{equation}
	i\frac{d}{dt}\left[ \langle k_{t}|\dot{m}_{t}\rangle \right] =\left[ \theta
	_{m}\left( t\right) -\theta _{k}\left( t\right) \right] \langle k_{t}|\dot{m}%
	_{t}\rangle \,\,\,\, (k \neq m) \,\, \text{ .}  \label{NonDiagSI}
	\end{equation}%
	By taking $\langle m_t |\dot{m}_{t}\rangle \equiv c_{mm}$ in Eq.~(\ref{DiagSI}), with $c_{mm}$ 
	denoting by hypothesis complex constants, we get 
	$\theta _{m}\left( t\right) =\theta _{m}\left(0\right) \equiv \theta_{m}$, namely, $\theta _{m}\left( t\right)$ is 
	a constant function $\forall m$. Moreover, by using 
	$\langle k_t |\dot{m}_{t}\rangle \equiv c_{km}$ in Eq.~(\ref{NonDiagSI}), with $c_{km}$ denoting 
	{\it nonvanishing} complex constants, we obtain $\theta _{k} = \theta _{m}$, $\forall k,m$. 
	If $c_{km}=0$, then $\theta _{k}$ and $\theta _{m}$ are not necessarily equal, but Eq.~(\ref{NonDiagSI}) 
	will also be satisfied by this choice. Therefore, it follows that $\theta_m(t)$ can be simply taken as
	\begin{equation}
	\theta_m(t) =\theta_m = \theta \,\,\, \forall m \, ,
	\end{equation} 
	with $\theta$ a single real constant. This concludes the proof. 
\end{proof}

	\section*{References}
\providecommand{\newblock}{}

	%\bibliography{mybib-noURL}
	%\bibliographystyle{J-Phys-A}

\end{document}